\documentclass[11pt]{article}

\usepackage{mathtools}
\usepackage{wrapfig}

\title{Improved Sample Complexity of Imitation Learning\\ for Barrier Model Predictive Control\thanks{The first two authors contributed equally. This work extends our previous result in \cite{pfrommer2024sample}, which has been accepted for publication in CDC 2024. An earlier version of this manuscript was submitted as part of DP's Master's thesis \cite{pfrommer2024samplethesis}.}
}

\author{Daniel Pfrommer\thanks{Massachusetts Institute of Technology. Email: \texttt{dpfrom@mit.edu}. }
 \and Swati Padmanabhan\thanks{Massachusetts Institute of Technology. Email: \texttt{pswt@mit.edu}.}
 \and Kwangjun Ahn\thanks{Microsoft Research. Email: \texttt{kwangjunahn@microsoft.com}.}
 \and Jack Umenberger\thanks{University of Oxford. Email: \texttt{jack.umenberger@eng.ox.ac.uk}.}
 \and Tobia Marcucci\thanks{Massachusetts Institute of Technology. Email: \texttt{tobiam@mit.edu}.}
 \and Zakaria Mhammedi\thanks{Massachusetts Institute of Technology. Email: \texttt{mhammedi@mit.edu}.}
 \and Ali Jadbabaie\thanks{Massachusetts Insitute of Technology. Email: \texttt{jadbabai@mit.edu}.}
 }
 
\usepackage[margin=1in]{geometry}

\usepackage[utf8]{inputenc}
\usepackage[T1]{fontenc}  
\usepackage{amsfonts, amsmath, amsthm, amssymb}       
\usepackage{mlmodern}
\usepackage{nicefrac}       
\usepackage{microtype}     
\usepackage[shortlabels]{enumitem}
\usepackage{fontawesome}

\usepackage[
backref=true,
natbib=true,
style=numeric,
giveninits=false, 
maxalphanames=10, 
maxcitenames=6,
maxbibnames=10,
sorting=ynt
]{biblatex}
\DefineBibliographyStrings{english}{
  backrefpage={cited on page},
  backrefpages={cited on pages}
}

\usepackage[usenames,dvipsnames]{xcolor}
\usepackage{hyperref}       
\usepackage{url}         

\definecolor{otherlightblue}{RGB}{0, 80, 250}
\definecolor{otherblue}{RGB}{0, 50, 100}
\definecolor{othergreen}{RGB}{60, 120, 0}

\hypersetup{
	colorlinks=true,
	pdfpagemode=UseNone,
	citecolor= magenta, 
	linkcolor= otherlightblue,
	urlcolor=magenta
}

\usepackage{thm-restate}
\numberwithin{equation}{section}

\usepackage{float}
\usepackage{algorithm, algorithmicx, algpseudocode}

\usepackage{makecell, booktabs, nicematrix} 
\usepackage{tcolorbox}
\usepackage[font=scriptsize]{caption}
\usepackage{graphicx, subcaption} 
\usepackage{enumitem}

\usepackage[capitalize, nameinlink]{cleveref}

\crefformat{none}{(#2#1#3)}
\crefname{equation}{Equation}{Equations}
\crefname{fact}{Fact}{Facts}
\crefname{lemmma}{Lemma}{Lemmas}
\crefname{figure}{Figure}{Figures}
\crefname{example}{Example}{Examples}
\crefname{defn}{Definition}{Definitions}
\crefname{ineq}{Inequality}{Inequalities}
\crefname{corollary}{Corollary}{Corollaries}
\creflabelformat{ineq}{#2{\upshape(#1)}#3} 
\crefalias{thm}{theorem}
\let\ref\cref

\newcommand{\compresslist}{ 
\setlength{\itemsep}{1pt}
\setlength{\parskip}{0pt}
\setlength{\parsep}{0pt}
}

\theoremstyle{plain}
\newtheorem{theorem}{Theorem}[section]

\newtheorem{claim}[theorem]{Claim}

\newtheorem{definition}[theorem]{Definition}
\newtheorem{fact}[theorem]{Fact}
\newtheorem{lemma}[theorem]{Lemma}
\newtheorem{example}[theorem]{Example}

\newtheorem{assumption}[theorem]{Assumption}
\newtheorem{problem}[theorem]{Problem}
\newtheorem{corollary}[theorem]{Corollary}

\newlist{propenum}{enumerate}{1} 
\setlist[propenum]{label=\alph*), ref=\theproposition(\alph*)}
\crefalias{propenumi}{proposition} 
\crefname{prop}{Proposition}{Propositions} 

\crefformat{equation}{(#2#1#3)}

\newlist{thmenum}{enumerate}{1} 
\setlist[thmenum]{label=(\roman*), ref=\thetheorem(\roman*)}
\crefalias{thmenumi}{theorem}
\crefname{theorem}{Theorem}{Theorems}
\newcommand{\bsigma}{\boldsymbol{\sigma}}
\newcommand{\nnz}{\mathrm{nnz}}
\newcommand{\vemod}{\widetilde{\mathrm{e}}}

\newcommand\numberthis{\addtocounter{equation}{1}\tag{\theequation}}  
\newlist{itemizec}{itemize}{2}
\setlist[itemizec,1]{label=\faCaretRight ,wide, parsep= 0.05pt, left = 15pt}
\newcommand{\nconstr}{m} 
\newcommand{\T}{T}
\newcommand{\R}{\mathbb{R}}
\newcommand{\Diag}{\mathrm{Diag}}

\newcommand{\fcl}[1]{f_{\mathrm{cl}}^{#1}}
\newcommand{\E}{\mathbb{E}}
\newcommand{\pirs}{{\boldsymbol{\pi}}^{\mathrm{rs}}}
\newcommand{\piexpert}{{\boldsymbol{\pi}}^{\star}} 
\newcommand{\pilearned}{\widehat{\boldsymbol{\pi}}}

\newcommand{\pimpc}{\boldsymbol{\pi}_{\mathrm{mpc}}}
\newcommand{\pibmpc}{\boldsymbol{\pi}_{\mathrm{mpc}}^\eta}

\DeclareMathOperator*{\argmin}{arg\,min}
\DeclareMathOperator*{\adj}{adj}

\newcommand{\calX}{\mathcal{X}}
\newcommand{\X}{\calX}

\newcommand{\calU}{\mathcal{U}}

\newcommand{\vueta}{\vec{u}^{\eta}}

\numberwithin{equation}{section}

\crefname{prob}{Problem}{Problems}
\crefname{ineq}{Inequality}{Inequalities}

\newcommand{\mat}[1]{{#1}}
\renewcommand{\vec}[1]{{#1}}

\newcommand{\mXt}{\widetilde{\mat{X}}}
\newcommand{\mI}{\mat{I}}
\newcommand{\mX}{\mat{X}}
\newcommand{\mLambda}{\mat{\Lambda}}

\newcommand{\mA}{\mat{A}}
\newcommand{\mB}{\mat{B}}
\newcommand{\mD}{\mat{D}}
 
\newcommand{\mQ}{\mat{Q}}
\newcommand{\mL}{\mat{L}}
\newcommand{\mR}{\mat{R}}
\newcommand{\mH}{\mat{H}}
\newcommand{\mF}{\mat{F}}
\newcommand{\mS}{\mat{S}}
\newcommand{\mG}{\mat{G}}
\newcommand{\mP}{\mat{P}}
\newcommand{\vw}{\vec{w}}

\newcommand{\mM}{\mat{M}}
\newcommand{\mDmod}{\widetilde{\mD}}
\newcommand{\vbmod}{\widetilde{\vb}}

\newcommand{\DuetaDx}{\frac{\partial \vueta}{\partial \vec{x}_0}}

\newcommand{\va}{\vec{a}}
\newcommand{\vd}{\vec{d}}
\newcommand{\vb}{\vec{b}}
\newcommand{\vx}{\vec{x}}
\newcommand{\vh}{\vec{h}}
\newcommand{\vu}{\vec{u}}
\newcommand{\vv}{\vec{v}}
\newcommand{\vy}{\vec{y}}
\newcommand{\ve}{\vec{e}}

\newcommand{\vg}{\vec{g}}
\newcommand{\vz}{\vec{z}}
\newcommand{\mK}{\mat{K}}
\newcommand{\mC}{\mat{C}}
\newcommand{\mU}{\mat{U}}
\newcommand{\mV}{\mat{V}}
\newcommand{\vxt}{\vx_t}
\newcommand{\zero}{\mathbf{0}}
\newcommand{\calK}{\mathcal{K}}
\newcommand{\calB}{\mathcal{B}}
\newcommand{\calV}{\mathcal{V}}
\newcommand{\calQ}{\mathcal{Q}}
\renewcommand{\det}{\mathrm{det}}
\newcommand{\mAmod}{\widetilde{\mA}}
\newcommand{\philb}{\textrm{res}_{\textrm{$\ell$.b.}}}
\newcommand{\phiK}{\phi_{\calK}}
\global\long\def\vc{\vec{c}}%
\global\long\def\vx{\vec{x}}%
\global\long\def\gap{\mathrm{gap}}%
\global\long\def\u{\vec{u}}%
\global\long\def\ue{\u_{\eta}}%
\global\long\def\v{\vec{v}}%
\global\long\def\vstar{\v^{\star}}%
\global\long\def\veta{\v_{\eta}}%
\global\long\def\us{\u^{\star}}%
\global\long\def\re{\textrm{res}}%

\newcommand{\phitilde}{\widetilde{\phi}}

\begin{document}
\maketitle
\begin{abstract}
Recent work in imitation learning has shown that having an expert controller that is both suitably smooth and stable enables stronger guarantees on the performance of the  learned controller. However, constructing such smoothed expert controllers for arbitrary systems remains challenging, especially in the presence of input and state constraints. As our primary contribution, we show how such a smoothed expert can be designed for a general class of systems using a log-barrier-based relaxation of a standard Model Predictive Control (MPC) optimization problem. 

Improving upon our previous work, we show that barrier MPC achieves theoretically optimal error-to-smoothness tradeoff along some direction. 
At the core of this theoretical guarantee on  smoothness is an improved lower bound we prove on the optimality gap of the analytic center associated with a convex Lipschitz function, which we believe could be of independent interest. We validate our theoretical findings via experiments, demonstrating the merits of our smoothing approach over randomized smoothing.

\end{abstract}

\newpage

\section{Introduction}\label{sec:introduction}
Imitation learning has emerged as a powerful tool in machine learning, enabling agents to learn complex behaviors  by imitating expert demonstrations  acquired either from a human demonstrator or a policy computed offline~\cite{pomerleau1988alvinn, ratliff2009learning, abbeel2010autonomous, ross2011reduction}. Despite its significant success, imitation learning often suffers from a {compounding error problem}: Successive evaluations of the approximate policy could accumulate error, resulting in out-of-distribution failures
\cite{pomerleau1988alvinn}. Recent results in imitation learning \cite{pfrommer2022tasil, tu2022sample, block2023provable} have identified \emph{smoothness} (i.e., Lipschitzness of the derivative of the optimal controller with respect to the initial state) and \emph{stability} of the expert as two key properties that circumvent this issue, thereby allowing for end-to-end performance guarantees for the final learned controller.

In this paper, our focus is on enabling such guarantees when the expert being imitated is a Model Predictive Controller (MPC), a powerful class of control algorithms based on solving an optimization problem over a receding prediction horizon~\cite{borrelli2017predictive}.
In some cases, the solution to this multiparametric optimization problem, known as the explicit MPC representation \cite{bemporad2002explicit},
can be pre-computed. For instance, in our setup --- linear systems with polytopic constraints --- the optimal control input is a piecewise affine  (and, hence, highly non-smooth) function of the state~\cite{bemporad2002explicit}.
However, the number of these pieces may grow exponentially with the time horizon and the state and input dimension, which  makes pre-computing and storing such a representation  impractical in high dimensions.

While the approximation of a linear MPC controller has garnered significant attention~\cite{chen2018approximating, maddalena2020neural, ahn2023model}, these prior works are primarily concerned with approximating the non-smooth explicit MPC using a neural network and then introducing schemes for enforcing the stability of the learned policy. In contrast, in our paper, we first construct a smoothed version of the expert and then apply theoretical results  derived from the imitation of a smoothed expert.

In particular, we demonstrate --- both theoretically and empirically ---  that a log-barrier  formulation of the underlying MPC optimization yields  smoothness properties similar to its randomized-smoothing-based counterpart, while being faster to compute. Similar to prior works~\cite{wills2004barrier, feller2013barrier, feller2014barrier}, our barrier MPC formulation replaces the constraints in the MPC optimization problem by a log-barrier in the objective (cf. \Cref{sec:barrier_mpc_all}). 
We show that, in conjunction with a black-box imitation learning algorithm, this provides end-to-end guarantees on the performance of the learned policy.

\paragraph{Our Contributions.} 
It is known from classical optimization theory~\cite{beck2012smoothing} that any smooth approximation of a nonsmooth function which is $O(\epsilon)$ close everywhere must have a smoothness constant (Lipschitzness of the gradient) at least $O(1/\epsilon)$. The well-known randomized smoothing technique~\cite{duchi2012randomized} (convolution with a smoothing kernel) is optimal in this sense; however, it does not preserve the stability properties of the underlying controller and hence is not well-suited for controls applications. 

Our main result is that log-barrier-based MPC~\cite{wills2004barrier} is an optimal smoother along some direction and outperforms randomized smoothing for controls tasks. More formally, for a given MPC,  letting $\vu^\star$ be the solution of the explicit MPC and $\vueta$ be the solution of the barrier-MPC formulation, with $\eta$ being the weight on the barrier, our main contributions for barrier MPC are as follows. 

We provide in \Cref{thm:hess_ueta_bounded} an upper bound of $O(\tfrac{1}{\sqrt{\eta + d^2} - d)})$ on the spectral norm of the Hessian of $\vueta$ with respect to $\vx_0$, where $d$ is the distance of the unconstrained solution from the polytope boundary under the appropriate metric.
Separately, we show that there exists a direction $\va$ (independent of $\eta$) along which the error $\va^\top(\vueta - \vu^\star)$ is at most  $O(\sqrt{\eta + d^2} - d)$. 
These two results together show that the controller smoothness and the error match along this direction $\va$, from which we infer that barrier MPC is an optimal smoother along this direction.
 Along the way, we show (\Cref{thm:convex_combination}) that the Jacobian of the log-barrier solution can be written as a convex combination of the Jacobian of the solution of the explicit MPC. In particular, this shows that the rate of change of $\vueta$  with respect to $\vx_0$ is bounded independent of the weight $\eta$ applied to the log barrier.  
 We also show (\Cref{thm:error_bound_barrier_mpc}) that overall,  the distance of $\vueta$ from $\vu^\star$ is bounded by $O(\sqrt{\eta})$.  
 Finally, we demonstrate through numerical experiments that barrier MPC outperforms randomized smoothing, thus empirically affirming the merits of controls-aware smoothing techniques.

A crucial technical component in obtaining the aforementioned  bound on the controller smoothness is a \textit{lower} bound on the distance of $\vueta$ from the boundary of the polytope (equivalently, a lower bound on the
optimality gap of the analytic center associated with a convex Lipschitz function). Intuitively, the nature of the self-concordant barrier already suggests that the solution to a problem with such a  barrier in the objective cannot be too close to the boundary of the constraint set. However, obtaining the desired upper bound on the controller smoothness requires an \emph{explicit quantification} of this distance. 
We provide (\Cref{thm:quad_opt_result}) such a bound for general convex Lipschitz functions via a novel reduction to the setting of linear objectives and then invoking a result by  \citet{zong2023short}. 
Furthermore, our smoothing analysis demonstrates that our lower bound is tight up to constants. We believe this result could be broadly useful to the optimization community.

\section{Problem Setup and Background}
\label{sec:notation}
\looseness=-1We first state our notation and setup that we use throughout. The notation $\|{}\cdot{}\|$ refers to the $\ell_2$ norm $\|{}\cdot{}\|_2$ for vectors and,  by extension, to the spectral norm (largest singular value) for square matrices. For a positive definite matrix $H$, we denote the local inner product $\|x\|_H = \sqrt{x^\top H x}$.  
Unless transposed, all vectors are column vectors. We use uppercase  letters for matrices and lowercase  letters for vectors. 
For a vector $\vx$, we use $\Diag(\vx)$ for the diagonal matrix with the entries of $\vx$ along its diagonal. We use $[n]$ for the set $\{1, 2, \dots, n\}$. Given a matrix $\mM \in \R^{n \times n}$ and $\sigma \in \{0,1\}^n$, we denote by $[\mM]_{\bsigma}$ 
the principal submatrix of $\mM$ corresponding to the rows and columns $i$ for which $\sigma_{i}=1$. We  use $\mM_{\bsigma}^{-1}$  to denote the matrix obtained by first computing the inverse 
of the matrix $[\mM]_{\bsigma}$ and then appropriately padding it with zeroes so that the resulting matrix $\mM_{\bsigma}^{-1}$ has  the same size as $\mM$. Similarly, we define $\adj(\mM)_{\bsigma}$ to be the matrix obtained by first computing the adjugate (the transpose of the cofactor matrix) of $[\mM]_{\bsigma}$ and then appropriately padding it with zeroes so that $\adj(\mM)_{\bsigma}$ has the same size as $\mM$. Lastly, $\mathcal{O}(\cdot)$ denotes expressions where numerical constants have been suppressed.

\looseness=-1We use the same setup as in our previous work \cite{pfrommer2024sample} and consider constrained discrete-time linear dynamical systems of the form,
\begin{align}
\vx_{t+1} = \mA\vxt + \mB \vu_t, \quad \vxt \in \calX_t, \vu_t \in \calU_t, \label{eq:dynamics}
\end{align}
with state $\vxt \in \calX_t \subseteq \R^{d_x}$ and control-input $\vu_t \in \calU_t \subseteq \R^{d_u}$ indexed by time step $t$, and state and input maps $\mA \in \R^{d_x \times d_x}$, $\mB \in \R^{d_x \times d_u}$. The sets  $\calX_t$ and $\calU_t$, respectively, 
are the compact convex state and input constraint sets described by the polytopes $$\calX_t := \{\vx \in \R^{d_x}\,\mid\, \mA_{x_t} \vx \leq \vb_{x_t}\}, \quad \calU_t := \{\vu\in \R^{d_u} \,\mid\, \mA_{u_t} \vu \leq \vb_{u_t} \},$$ where $\mA_{x_t} \in \R^{k_x \times d_x}$, $\mA_{u_t} \in \R^{k_u \times d_u}$, $\vb_{x_t} \in \R^{k_x}$, and  $\vb_{u_t} \in \R^{k_u}$. We use $\mA_{x} \in \R^{(T \cdot k_x) \times d_x}, \mA_{u} \in \R^{(T \cdot k_u) \times d_u}, \vb_x \in \R^{T \cdot k_x}$, and  $\vb_u \in \R^{T \cdot k_u}$ to denote the vertically stacked constraints for the full sequences $x_{1:\T}$ and $u_{0:\T-1}$. 
A constraint $f(\vx)\leq 0$ is said to be ``active'' at $\vy$ if $f(\vy)=0$. Given a polytope $Ax \leq b$, we say that the quantity $b_i - a_i^\top x$ is its ``$i^\mathrm{th}$ residual''.  
For notational convenience, we overload $\phi$ to compactly denote the vector of constraint residuals for a state $\vx$ and input $\vu$ as well as for the sequences $\vx_{1:\T}$ and $\vu_{0:\T-1}$: 
\[
\phi_t(\vx_t, \vu_{t-1}) := \begin{bmatrix}\vb_{x_t} - \mA_{x_t} \vx_{t} \\ \vb_{u_{t-1}} - \mA_{u_{t-1}} \vu_{t-1}\end{bmatrix}, \quad
\phi(\vx_0,\vu_{0:\T-1}) := \begin{bmatrix}\phi_1(\vx_1, \vu_0) \\ \vdots \\ \phi_{\T}(\vx_\T, \vu_{\T-1})\end{bmatrix}.\numberthis\label{eq:phi}
\] 
We consider deterministic state-feedback control policies of the form $\pi: \calX \to \calU$ and denote the closed-loop system under $\pi$ by
$\fcl{\pi}(\vx) := \mA \vx + \mB\pi(\vx)$. 
We  use $\piexpert$ to refer to the expert policy and $\pilearned$ to refer to its learned approximation. 

\looseness=-1In particular, our
principal choice of $\piexpert$ 
in this paper is an MPC with quadratic cost and linear constraints. The MPC policy is obtained by solving the following minimization problem over future actions $\vu:= \vu_{0:\T-1}$
with quadratic cost in $\vu$ and states $\vx := \vx_{1:\T}$: 
\[ 
\begin{array}{ll}
\vspace{0.5em}
\mbox{minimize}_{\vu} & V(\vx_0, \vu):= \sum_{t=1}^\T \vxt^\top \mQ_t \vxt + \sum_{t=0}^{\T-1} \vu_t^\top \mR_t \vu_t \\
\mbox{subject to } &\vx_{t + 1} := \mA\vxt + \mB\vu_t,  \\
&\vx_\T \in \calX, \vu_0 \in \calU, \\
&\vxt \in \calX, \vu_t \in \calU, \; \forall t \in [\T-1],
\end{array}\label[prob]{eq:V}\numberthis
\] 
where $\mQ_t$ and $\mR_{t-1}$ are positive definite for all $t \in [\T]$.
For a given state $\vx$, the corresponding input $\pimpc$  of the MPC  is: 
\begin{align}\label[prob]{eq:pi_mpc}
\pimpc(\vx) := \argmin_{\vu_0} \min_{\vu_{1:\T-1}} V(\vx, \vu_{0:\T-1}),
\end{align}
where the minimization is  over  the feasible set defined in \Cref{eq:V}.
For $\pimpc$ to be well-defined, we assume that $V(\vx_0, \vu)$ has a unique global minimum in $\vu$ for all {feasible} $\vx_0$.

\subsection{Explicit Solution to MPC}\label{sec:explicitMPC}
As first noted by \citet{bemporad2002explicit}, explicit MPC rewrites  \Cref{eq:pi_mpc} as a multi-parametric quadratic program with linear inequality constraints and  solves it for every possible combination of active constraints, building an analytical solution to the control problem.
{Following this known derivation (see \citep[Section 4]{bemporad2002explicit} and \citep[Chapter 11]{borrelli2017predictive}),} we rewrite \Cref{eq:pi_mpc} as the optimization problem, in variable $\vu := \vu_{0:\T-1} \in \R^{\T  d_u}$, as described below: 
\[ 
\begin{array}{ll}
\vspace{0.5em}
\mbox{minimize}_{\vu} &\calV(\vx_0, \vu):= \tfrac{1}{2}\vu^\top \mH \vu  - \vx_0^\top \mF \vu\\
\mbox{subject to } &\mG \vu \leq \vw + \mP \vx_0,  
\end{array}\numberthis\label[prob]{eq:reformulated}
\] 
with cost matrices $\mH\in\R^{\T \cdot d_u \times T \cdot d_u}$ and $\mF \in \R^{d_x \times \T \cdot d_u}$ and $m$ constraints captured via  $\mG \in \R^{\nconstr \times \T \cdot d_u}$, and $\mP \in \R^{m \times d_x}$, and vector $\vw\in  \R^{\nconstr}$, all given by 
\begin{align*}
    \mH &= \mR_{0:\T-1} + \widehat{\mB}^\top \mQ_{1:\T} \widehat{\mB}, \quad \mF = -2\widehat{\mA}^\top \mQ_{1:\T}\widehat{\mB}, \\
    \mG &= \begin{bmatrix}
        \mA_u \\
        \mA_x\widehat{\mB}
    \end{bmatrix},\quad \mP = \begin{bmatrix}
        0 \\
        -\mA_x\widehat{\mA}
    \end{bmatrix}, \quad \vw = \begin{bmatrix}
        \vb_u \\
        \vb_x
    \end{bmatrix},
\end{align*}
where $\mQ_{1:\T}$ and $\mR_{0:\T-1}$ are block diagonal  with $\mQ_1, \dots, \mQ_\T$ and $\mR_0, \dots, \mR_{\T-1}$ on the diagonal, and $\widehat{\mB}$ and $\widehat{\mA}$ are,
\begin{align*}
    \widehat{\mA} = \begin{bmatrix}\mA 
    \\\mA^2 
    \\ \vdots \\ \mA^{\T}\end{bmatrix}, \quad
    \widehat{\mB} = \begin{bmatrix}
        \mB & 0 & \dots & 0 & \\
        \mA\mB & \mB & \dots & 0 \\
        \vdots & \vdots & \ddots & \vdots \\
        \mA^{\T -1}\mB & \mA^{\T-2}\mB & \dots & \mB
    \end{bmatrix}
\end{align*}
so that $\vx_{1:\T} = \widehat{\mA}\vx_0 + \widehat{\mB}\vu$. 
\begin{assumption}\label{assumption:ball_inside_outside}
We assume that the constraint polytope in \Cref{eq:reformulated} contains   a 
ball of radius $r$ and is contained inside an origin-centered ball of radius $R$.
\end{assumption}
We now state the solution to \Cref{eq:reformulated} and later (in \Cref{thm:convex_combination}) show how it appears in the smoothness of the \textit{barrier} MPC solution.

\begin{lemma}[{\citet[Theorem 2]{bemporad2002explicit}}]\label{thm:dudx}
{Given a feasible initial state $\vx$,} let $\sigma(\vx) \in \{0,1\}^\nconstr$ denote the indicator of active constraints of the optimizer of \Cref{eq:reformulated}, with $\sigma_i(\vx)= 1$ iff the $i$th constraint is active.
For $\bsigma \in \{0,1\}^\nconstr$, let $P_{\bsigma} = \{ \vx | \sigma(\vx) = \bsigma\}$ be the set of initial states $\vx$ for which the solution has active constraints determined by $\bsigma$. Then for $\vx_0 \in P_{\bsigma}$, the solution $\vu$ of \Cref{eq:reformulated} is expressed as
$\vu = K_{\sigma} \vx_0 + k_{\sigma}$, 
where $K_\sigma$ and $k_\sigma$ are defined as: 
\begin{equation}\numberthis\label{eq:K_sigma}
\begin{aligned}
    K_{\sigma} &:= \mH^{-1}[\mF^\top - \mG^\top(\mG \mH^{-1} \mG^\top)_{\bsigma}^{-1}(\mG \mH^{-1} \mF^\top - \mP)],\\
    k_\sigma &:= \mH^{-1} \mG^\top (\mG \mH^{-1} \mG^\top)_{\bsigma}^{-1} \vw.
\end{aligned}
\end{equation}
\end{lemma} 
Using this result, one may pre-compute an efficient lookup structure mapping $\vx \in P_\sigma$ to $K_\sigma, k_\sigma$. However, since every combination of active constraints may  yield a potentially unique feedback law, the number of pieces to be computed may grow \textit{exponentially} in the problem dimension or time horizon. For instance, even the simple two-dimensional toy system visualized in \Cref{fig:explicit_mpc} has  $261$ pieces. As a result, it might be computationally intractable to even merely enumerate or store all pieces of the explicit MPC in high dimensions or over long time horizons. 

\begin{figure}
\begin{center}
    \includegraphics[width=0.5\linewidth]{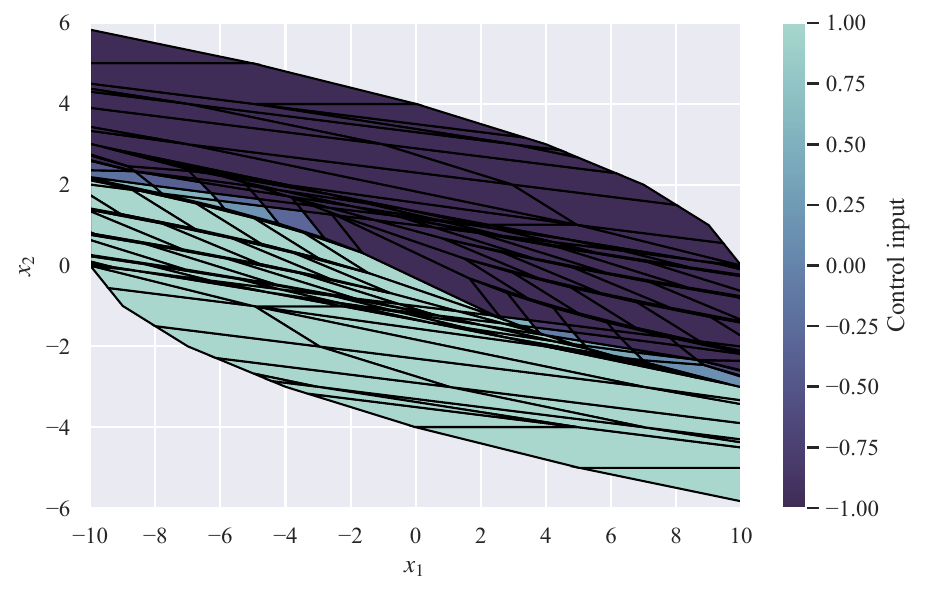}
\end{center}
    \vspace{-15pt}
    \caption{The explicit MPC controller for $A = \begin{bmatrix}1 & 1 \\ 0 & 1\end{bmatrix}, B = \begin{bmatrix}0 \\ 1\end{bmatrix}, Q = I, R = 0.01, \T=10$ with the constraints $\|x\|_\infty \leq 10, |u|\leq 1$. For this simple 2-dimensional system there are $261$ $K_\sigma$. This figure appeared in our previous work \cite{pfrommer2024sample}.}
    \label{fig:explicit_mpc}
\end{figure}

This observation motivates us to consider a learning-based approach. 
{In the spirit of imitation learning discussed in \Cref{sec:introduction}, we} approximate explicit MPC using a polynomial number of sample trajectories, collected offline. We introduce this framework in the next section.

\section{Motivating Smoothness: Imitation Learning Frameworks}\label{sec:learning_guarantees}
In this section, we instantiate  the {imitation learning} framework to motivate our approach {of barrier MPC}.
{We use the Taylor series based imitation learning framework introduced by  \citet{pfrommer2022tasil}, which gives high-probability guarantees on the quality of an approximation.}

\subsection{Taylor Series Imitation Learning}
{We first introduce the setting for imitation learning.}
Suppose we are given an expert controller $\piexpert$, a policy class $\Pi$, a distribution of initial conditions $\mathcal{D}$, and $N$ sample trajectories $\{\vx_{0:K-1}^{(i)}\}_{i=1}^N$ of length $K$, with $\{\vx_0^{(i)}\}_{i=1}^N$ sampled i.i.d from $\mathcal{D}$. As formalized in \Cref{prop:goodness_of_learned_policy}, our goal is to find an approximate policy $\pilearned \in \Pi$ such that given a suitably small accuracy parameter $\epsilon$, the closed-loop states $\widehat{\vx}_t$  and $\vx_t^\star$ induced by $\pilearned$ and $\piexpert$, respectively, satisfy, with high probability over $\vx_0 \sim \mathcal{D}$,
$$\|\widehat{\vx}_t - \vx^\star_t\| \leq \epsilon, \forall t > 0.$$ To understand the sufficient conditions for such a guarantee, we now introduce a few definitions.
We first assume through  \Cref{assumption:closeness_approx_expert} that $\pilearned$ has been chosen by a black-box supervised imitation learning algorithm which, given the input data, produces a $\pilearned \in \Pi$ such that, with high probability over the  distribution induced by $\mathcal{D}$, the policy {and its Jacobian} are close to the expert.

\begin{assumption}\label{assumption:closeness_approx_expert}
For some $\delta \in (0,1), \epsilon_0 > 0, \epsilon_1 > 0$ and given $N$ trajectories $\{\vx_{0:K-1}^{(i)}\}_{i=1}^{(N)}$ of length $K$ sampled i.i.d. from $\mathcal{D}$ and rolled out under $\piexpert$, the approximating policy $\pilearned$ satisfies:
\begin{align*}
    \mathbb{P}_{\vx_0 \sim \mathcal{D}}\bigg[&\sup_{k \geq 0}\|\pilearned(\vx_k) - \piexpert(\vx_k)\| \leq \epsilon_0/N \,\,\, \wedge \,\,\,  \sup_{k \geq 0}\left\|\frac{\partial\pilearned}{\partial \vx}(\vx_k) - \frac{\partial \piexpert}{\partial \vx}(\vx_k)\right\| \leq \epsilon_1/N \bigg] \geq 1 - \delta. 
\end{align*}\label{assum:bounds}
\end{assumption}
As shown in \cite{pfrommer2022tasil}, an example in which  \Cref{assumption:closeness_approx_expert} holds is when $\pilearned$ is chosen as an empirical risk minimizer from a class of twice differentiable parametric functions with $\ell_2$-bounded parameters, e.g., dense neural networks with smooth activation functions and trained with $\ell_2$ weight regularization.  We refer the reader to \cite{pfrommer2022tasil, tu2022sample} for other valid examples of $\Pi$. 
Further, note that the above definition   requires  generalization  on only the state distribution induced by the expert, rather than on the distribution induced by the learned policy, as is the case in \cite{chen2018approximating, ahn2023model}.

Next, we define a weaker variant of the standard \emph{incremental input-to-state stability} ($\delta$ISS) \cite{vosswinkel2020determining} and assume, in \Cref{assum:stable}, that this property holds for the expert policy. 

\begin{definition}[Local Input-to-State Stability with Linear Gain, cf. \cite{pfrommer2022tasil}]\label{def:locIncStab} For all initial conditions $\vx_0 \in \calX$ and bounded sequences of input perturbations $\{\vec{ \Delta}_t\}_{t > 0}$ that satisfy $\|\vec{\Delta}_t\| < \kappa$, let $\overline{\vx}_{t+1} = \fcl{\pi}(\overline{\vx}_t, 0)$, $\overline{\vx}_0 = \vx_0$ be the nominal trajectory, and let $\vx_{t+1} = \fcl{\pi}(\vx_t, \vec{\Delta}_t)$ be the perturbed trajectory. We say that the closed-loop dynamics under $\pi$ is $(\kappa, \gamma)$-locally-incrementally input-to-state stable for linear gain $\gamma$ if $\kappa, \gamma > 0$ and,
\begin{align*}
    \|\vx_t - \overline{\vx}_t\| \leq \gamma \cdot \max_{k < t} \|\vec{\Delta}_k\|, \quad \forall t \geq 0.
\end{align*}
\end{definition}

\begin{assumption}\label{assum:stable}
    The expert policy $\piexpert$ is ($\kappa,\gamma$)-locally incrementally input-to-state stable.
\end{assumption}
As noted in \cite{pfrommer2022tasil}, local incremental input-to-state stability (local $\delta$ISS) is a much weaker criterion than  regular incremental input-to-state stability. We will later show in \Cref{lem:iss_diss} that under mild assumptions even input-to-state stabilizing (ISS)  policies (defined in (\Cref{def:inp_stab})) are  locally $\delta$ISS. There is considerable prior work (see, e.g., \cite{zamani2011lyapunov, pouilly2020stability}) demonstrating that ISS holds under mild conditions for both the explicit MPC and the barrier-based MPC under consideration in this paper. Putting these facts together then implies local $\delta$ISS of barrier MPC.  
Having established some preliminaries for stability, we now move on to the smoothness property we consider.  
\begin{definition}[\cite{pfrommer2022tasil}]\label{def:def_smoothness} We say that an MPC policy $\pi$ is $(L_0, L_1)$-smooth if for all $\vx, \vy \in \X$,
\begin{align*}
    \|\pi(\vx) - \pi(\vy)\| &\leq L_0\|\vx - \vy\|\quad \text{ and } \quad
    \left\|\frac{\partial \pi}{\partial x}(\vx) - \frac{\partial \pi}{\partial x}(\vy)\right\| \leq L_1\|\vx - \vy\|.
\end{align*}
\end{definition}
\begin{assumption}\label{assumption:smoothness_of_exp_and_learned} The expert policy $\piexpert$ and the learned policy $\pilearned$ are both $(L_0, L_1)$-smooth.
\end{assumption}
At a high level,  assuming smoothness of the expert and the learned policy helps implicitly ensure that the learned policy captures the stability of the expert in a neighborhood around the data distribution. If the expert or learned policy were to be only piecewise smooth (as is the case, e.g., with standard MPC-based solution of LQR), a transition from one piece to another in the expert not replicated by the learned policy could  result in unstable closed-loop behavior. 

Having stated all the necessary assumptions, we are now ready to state below  the main export of this section, given by \cite{pfrommer2022tasil}, guaranteeing closeness of the learned and expert policies. 

\begin{fact}[\cite{pfrommer2022tasil}, Corollary A.1]\label{prop:goodness_of_learned_policy} Provided $\piexpert, \pilearned$ are $(L_0,L_1)$-smooth, $(\kappa,\gamma)$-locally incrementally stable, and $\widehat{\pi}$ satisfies \Cref{assum:bounds} with $\frac{\epsilon_0}{N} \leq \min\{\frac{1}{16\gamma^2 L_1}, \frac{1}{16\gamma}, \frac{\kappa}{8\gamma}\}$ and $\frac{\epsilon_1}{N} \leq \frac{1}{4\gamma}$, $\delta > 0$, then with probability $1-\delta$ for $x_0 \sim \mathcal{D}$, we have 
\begin{align*}
    \|\widehat{\vx}_t - \vx^\star_t\| \leq \frac{8 \gamma \epsilon_0}{N} \quad \forall t \geq 0.
\end{align*}
\end{fact}
The upshot of this result is that provided the MPC policy  $\piexpert$ is $(L_0, L_1)$-smooth, to match the trajectory of $\piexpert$ with high probability,  {we  need to match the Jacobian and value of $\piexpert$ on \emph{only} $N  K$ pieces.} 
This is in contrast to prior work such as \cite{maddalena2020neural, karg2020efficient, chen2018approximating} on approximating explicit MPC, which require sampling new control inputs during training (in a reinforcement learning-like fashion) or post-training verification of the stability properties of the network.

However, as we noted in \Cref{sec:introduction}, these strong guarantees crucially require a smooth expert controller. We investigate two approaches for smoothing $\pimpc$: randomized smoothing and barrier MPC. Before doing so, we first consider what constitutes an ``optimal'' smoothing approach in terms of the smallest possible Hessian norm for a given level of approximation error.

\subsection{An Overview of Optimal Smoothing}

We begin by considering the properties of a general smoothing algorithm. For simplicity, in this section, we consider smoothing functions of the form $f: \R \to \R$, although we note that  this analysis can easily be extended to $f: \R^n \to \R^m$ by considering arbitrary paths $\R \to \R^n$ and projection $\R^m \to \R$. This motivates the following definition of a smoothing algorithm.

\begin{definition}[$\epsilon$-Smoothing Algorithm] Let $\epsilon > 0$. An $\epsilon$-smoothing algorithm $\mathcal{S}$ for a function class $\mathcal{F}$ is a map $\mathcal{S}: \mathcal{F} \to C^1$ where $C^1$ is the class of functions $\R \to \R$ with continuous derivatives. Furthermore $\mathcal{S}$ satisfies, 
\begin{align*}
    \sup_{x} \left\|\mathcal{S}(f(x)) - f(x)\right\| \leq \epsilon \quad \forall f \in \mathcal{F}.
\end{align*}
\end{definition}
\noindent Analogously, we define a general smoothing algorithm which can smooth functions for arbitrary $\epsilon$.
\begin{definition}[Smoothing Algorithm] A general smoothing algorithm $\mathcal{S}$ for a function class $\mathcal{F}$ is a map $\mathcal{S}: \R \times \mathcal{F} \to C^1$ where $\mathcal{S}(\epsilon, \cdot)$ is an $\epsilon$-smoothing algorithm.
\end{definition}

It is known that for any $\epsilon$-smoothing algorithm $\mathcal{S}$ for the class of $L$-Lipschitz functions (which we denote by $\mathcal{L}_L$), there exists $f \in \mathcal{L}_L$ such that the derivative of $g := \mathcal{S}(f)$ has Lipschitz constant at least $\mathcal{O}(\frac{1}{\epsilon})$. For in-depth treatment of the subject under a more general setting, we direct the reader to \citet{kornowski2021oracle} and \citet{beck2012smoothing}. A simple example of such a function for which this bound holds is the scaled absolute value function $x \to C|x|$.

\begin{lemma}[\cite{kornowski2021oracle}, Lemma 30]\label{lem:hes_bound_abs} Let $\mathcal{S}$ be any $\epsilon$-smoothing algorithm $\mathcal{S}: \mathcal{L}_L \to C^1$ for $L, \epsilon > 0$ and let $f(x) := L|x|$, $g(x) := \mathcal{S}(f)$. Then there exists $x, y \in \R$ such that,
\begin{align*}
    |\nabla g(x) - \nabla g(y)| \geq \frac{L^2}{9\epsilon}|x - y|.
\end{align*}
\end{lemma}

\begin{proof}Consider the value of $g(x)$ at $x = -\frac{3\epsilon}{L}$, $x= 0$, and $x= \frac{3\epsilon}{L}$. Since $\mathcal{S}$ is an $\epsilon$-smoothing algorithm and $f(-\frac{3\epsilon}{L}) = f(\frac{3\epsilon}{L}) = 3\epsilon$ and $f(0) = 0$, we can conclude that $g(-\frac{3\epsilon}{L}), g(\frac{3\epsilon}{L}) > 2\epsilon$, $g(0) < \epsilon$. This implies that $g(\frac{3\epsilon}{L}) - g(0) \geq \epsilon$ and $g(0) - g(-\frac{3\epsilon}{L}) \leq -\epsilon$. By the mean value theorem, there exist points $y \in [-\frac{3\epsilon}{L}, 0]$ and $x \in [0, \frac{3\epsilon}{L}]$ such that,
\begin{align*}
    \nabla g(y) < - \frac{L}{3}, \nabla g(x) > \frac{L}{3}.
\end{align*}
Note that $|x - y| \leq \frac{6\epsilon}{L}$ or that $\frac{L^2}{9\epsilon}|x - y| \leq \frac{2L}{3}$. Therefore, we may complete the proof by  noting,
\begin{align*}
    |\nabla g(x) - \nabla g(y)| \geq \frac{2L}{3} \geq \frac{L^2}{9\epsilon}|x -y|.
\end{align*}
\end{proof}

The above result suggests an inherent tradeoff between the approximation error $\epsilon$ and the Lipschitzness of the derivative of the smoothed function. Using intuition from the above result, we now state a more general result for arbitrary piecewise twice differentiable functions where the derivatives at the  boundaries of the pieces do not necessarily match.

\begin{lemma}\label{thm:best_smooth_bound}Let $f: \R \to \R$ be a function with piecewise continuous derivatives. Let $c \in \R$ be a point such that $\lim_{x \to c^-} \nabla f(x) = a$ and $\lim_{x \to c^+} \nabla f(x) = b$ where $a \neq b$ (i.e. the derivative is discontinuous). Then for sufficiently small $\epsilon$ and any $\epsilon$-smoothing algorithm $\mathcal{S}$, we have that for $g := \mathcal{S}(f)$ there exist $x, y$ such that,
\begin{align*}
    |\nabla g(x) - \nabla g(y)| \geq \frac{|a - b|^2}{144\epsilon}|x - y|.
\end{align*}
\end{lemma}
\begin{proof}Without loss of generality, we can shift $f$ so that $c = 0, f(0) = 0$. Similarly, we can also subtract off $\frac{(a + b)}{2}x$ from both $f$ and $g$ as well as transform $f(x) \to f(-x), g(x) \to g(-x)$ such that $\lim_{x \to 0^-} \nabla f(x) = -\frac{|a - b|}{2}, \lim_{x \to 0^+} \nabla f(x) = \frac{|a - b|}{2}$. Let $d := \frac{|a - b|}{2}$.

Since $f$ has piecewise continuous derivative, by definition there exists some radius $\delta > 0$ around $0$ such that $f$ is differentiable on $(-\delta, 0)$ and $(0, \delta)$ and that $\nabla f(x) \leq -\frac{d}{2}$ for $x \in (-\delta,0)$ and $\nabla f(x) \geq \frac{d}{2}$ for $x \in (0, \delta)$. We can therefore lower bound,
\begin{align*}
    f(x) \geq -\frac{d}{2}x \quad \forall x \in (-\delta, 0), \quad \quad f(x) \geq \frac{d}{2}x \quad \forall x \in (0,\delta).
\end{align*}
Similar to \Cref{lem:hes_bound_abs}, we note that for $\epsilon \leq \frac{d}{6}\delta$, $f(\frac{6\epsilon}{d}), f(\frac{6\epsilon}{d}) > 3\epsilon$. Since $f(0) = 0$, $g(0) \leq \epsilon$ and therefore $g(0) - g(\frac{6\epsilon}{d}) < -\epsilon, g(\frac{6\epsilon}{d}) - g(0) \geq \epsilon$. Therefore for some $x, y \in \R$ such that $|x - y| \leq \frac{12\epsilon}{d}$, we have that,
\begin{align*}
    |\nabla g(x) - \nabla g(y)| \geq \frac{d}{3} \geq \frac{d^2}{36\epsilon}|x - y| = \frac{|a - b|^2}{144\epsilon}|x - y|.
\end{align*}
\end{proof}

The above result suggests that the derivative of an $\epsilon$-smoothed function has a Lipschitz constant lower bounded by the square of the ``discontinuity'' in the derivatives times the inverse of the largest approximation error. If the smoothed function is twice differentiable, this is equivalent to a lower bound on the Hessian.

Guided by the above results, we now state our definition for an ``optimal smoothing'' algorithm, which is a smoothing algorithm such that the above bound is tight, up to a constant. For simplicity, we define optimal smoothing only for $L$-Lipschitz functions.
\begin{definition}\label{def:optimal_smooth}A smoothing algorithm $\mathcal{S}: \mathcal{R} \times \mathcal{F} \to C^1$ for a function class $\mathcal{F}$ is worst-case optimal up to a constant if there exists $C > 0$ such that, for any sufficiently small $\epsilon > 0$, $L > 0$, and $L$-Lipschitz function $f \in \mathcal{L}_L \subset \mathcal{F}$, the following inequality holds with $g := \mathcal{S}(\epsilon, f)$,
\begin{align*}
    \|\nabla g(x) - \nabla g(y)\| \leq C\frac{L^2}{\epsilon} \|x - y\|.
\end{align*}
\end{definition}
Note that, by the above lemmas, an algorithm satisfying \Cref{def:optimal_smooth} yields smoothed functions where the bound on the Hessian  is at most a constant factor worse than the best possible bound for Lipschitz functions. Since the explicit MPC is always Lipschitz, for our purposes we will simply refer to smoothing algorithms satisfying \Cref{def:optimal_smooth} as ``optimal smoothers''.

In the next two sections, we answer the question of whether an optimal smoothing algorithm can preserve the stability of an explicit MPC controller. We will show that while randomized smoothing is an optimal smoother, there exist systems for which randomized smoothing does not preserve the stability of the system. We will then introduce barrier MPC and prove that barrier MPC is an optimal smoother along a certain direction.

\begin{figure*}
    \centering
    \includegraphics[width=1\linewidth]{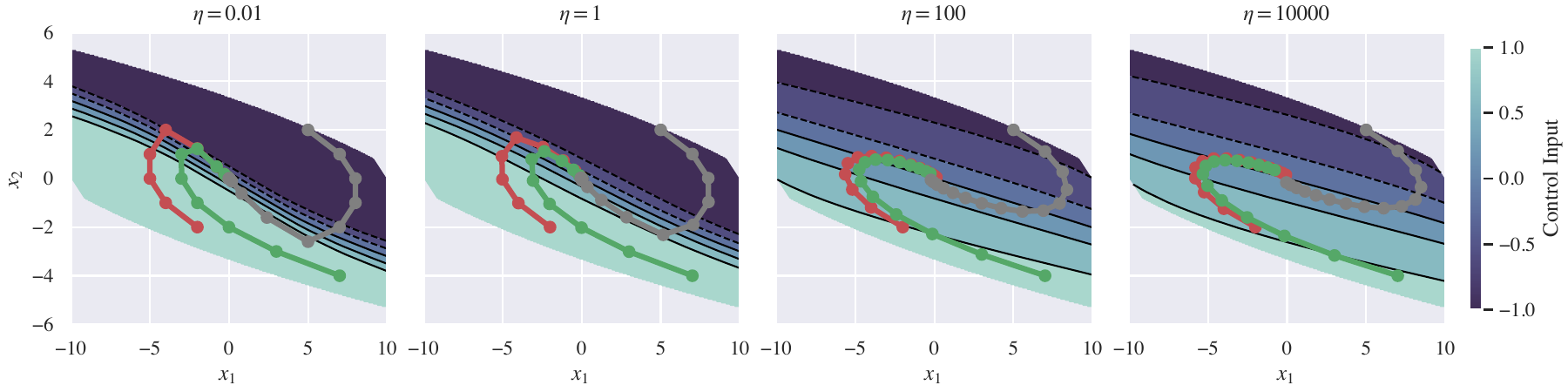}
    \caption{Visualizations of the log-barrier MPC control policy and several trajectories for the same system as \Cref{fig:explicit_mpc} and different choices of $\eta$. This figure appeared in our previous work \cite{pfrommer2024sample}.}
    \label{fig:smoothing_contours}
\end{figure*}

\subsection{First Approach: Randomized Smoothing}\label{sec:randomized_smoothing}
We first consider randomized smoothing (see, {e.g.}, \cite{duchi2012randomized}) as a baseline approach for smoothing the expert policy $\piexpert$.  Here, the imitator is learned with a loss function that randomly samples with 
noise drawn from a chosen probability distribution in order to smooth the policy, effectively convolving the controller with a smoothing kernel. This approach corresponds to the following controller.
\begin{definition}[Randomized Smoothed MPC]\label{def:rand_smoo} Given a control policy $\pimpc$ of the form \cref{eq:pi_mpc}, a desired zero-mean noise distribution $\mathcal{P}$, and a smoothing parameter $\sigma > 0$, the randomized-smoothing based MPC, $\pirs$, is defined as:
\begin{align*}
    \pirs(\vx) &:= \E_{\vw \sim \mathcal{P}}[\pimpc(\vx+ \sigma\vw)].
\end{align*}
\end{definition}

The distribution $\mathcal{P}$ in \Cref{def:rand_smoo} is usually chosen such that the following guarantees on error and smoothness hold. 

\begin{fact}[{\citet[Appendix E]{duchi2012randomized}}] \label{thm:randomized_bounds}
For control policy $\pimpc: \calX \to \calU$, $\calX \subset \R^{d_x}, \calU \subset \R^{d_u}$ and  $\mathcal{P} \in  \{\mathrm{Unif}(B_{\ell_2}(1)), \mathrm{Unif}(B_{\ell_\infty}(1)), \,\mathcal{N}(0,I)\}$,  there exist constants $C_0, C_1>0$ that depend on $d_x$. Let $L$ be the Lipschitz constant of the given control policy $\pimpc$. Then, for any smoothing parameter $\sigma > 0$, the associated $\pirs$ satisfies,
\begin{align*}
    \|\pirs(\vx) - \pimpc(\vx)\| &\leq C_0 \sigma &\forall \vx \in \X, \\
    \|\nabla \pirs(\vx) - \nabla \pirs(\vy)\| &\leq \frac{C_1 L^2}{\sigma}\|\vx - \vy\|  &\forall \vx,\vy \in \X.
\end{align*} 
This implies that randomized smoothing is an optimal smoother for the given choices of $\mathcal{P}$.
\end{fact}

However, using randomized smoothing to obtain a smoothed policy has three key disadvantages. First, $ \E_{\vw \sim \mathcal{P}}[\pimpc(\vx+ \epsilon \vw)]$ is evaluated via sampling, which means the expert policy must be continuously re-evaluated during training in order to guarantee convergence to the smoothed policy. Secondly, smoothing in this manner may cause $\pirs$ to violate state constraints. Finally, simply smoothing the policy may not preserve the stability of $\pimpc$. 
The first two stated problems arise due to randomized smoothing \emph{oversmoothing} the underlying controller. Consider the following example.
\begin{example}Consider the system $f(x_t,u_t) = 2x_t + u_t$ and controller $\pi^\star(x) = \min(\max(-2x, -1), 1)$. We can see that as $\sigma \to \infty$ (where $\sigma$ is the smoothing parameter from \Cref{def:rand_smoo}), we have $\pi^{\mathrm{rs}}(x) \to 0$ for all $x$.
\end{example}
The above example shows that $\pi^{\mathrm{rs}}$ is not stable for high $\sigma$. Ideally, we would like to  aggressively smooth only the discontinuities that do not affect stability or constraint guarantees. This requires a smoothing technique that is aware of when more aggressive control inputs are being taken in order to more quickly stabilize versus preserve some constraint guarantees.
As we shall show,  barrier MPC is precisely one such method that also preserves state guarantees. We now define barrier MPC and bound  the approximation error and  smoothness of the resulting controller.

\section{Our Approach to Smoothing: Barrier MPC}\label{sec:barrier_mpc_all}
Having described the  guarantees obtained via randomized smoothing, we now consider smoothing via self-concordant {barrier} functions, a notion introduced by \citet{nesterov1994interior}. 

\begin{definition}[\citet{nesterov1994interior}]\label{def:sc_and_scb}
A convex, thrice differentiable function $\phi:\calQ\mapsto\R$ is a $\nu$-self-concordant barrier on an open convex set $\calQ\subseteq \R^n$ if the following conditions hold. 
\begin{enumerate}[label=(\roman*)]
\compresslist{
\item For all sequences $\vx_i\in \calQ$ converging to the boundary of $\calQ$, we have  $\lim_{i\rightarrow\infty}\phi(\vx_i)\rightarrow\infty$.
\item For all  $\vx\in \calQ$ and  $\vh\in \R^n$, we have the bound $\lvert \mathcal{D}^3 f(\vx)[\vh, \vh, \vh]\rvert \leq 2(\mathcal{D}^2 f(\vx)[\vh, \vh])^{3/2}, $ where $\mathcal{D}^k f(\vx)[\vh_1, \dotsc, \vh_k]$ is the $k$-th derivative of $f$ at $\vx$ along  directions $\vh_1,\dotsc,\vh_k$,
\item For all $\vx\in\calQ$, we have $\nabla f(\vx)^\top (\nabla^2 f(\vx))^{-1} \nabla f(\vx) \leq \nu$. The parameter $\nu$ satisfies $\nu\geq 1.$
}
\end{enumerate}
\end{definition}
The self-concordance property essentially says that locally, the Hessian does not change too fast --- it has therefore proven extremely useful in interior-point methods to design fast algorithms for (constrained) convex programming~\cite{zbMATH03301975, karmarkar1984new} and has also found use in model-predictive control \cite{wills2004barrier, feller2013barrier, feller2014barrier, feller2015weight, feller2016relaxed} in order to ensure strict feasibility of the control inputs. 

In this paper, we consider {barrier MPC} as a naturally smooth alternative to randomized smoothing of \Cref{eq:pi_mpc}. In barrier  MPC, the inequality constraints occurring in the  optimal control
problem are eliminated by incorporating them into the cost
function via suitably scaled barrier terms. We work only with the log-barrier, which turns a  constraint $f(\vx)\geq 0$ into the term $-\eta \log(f(\vx))$ in the minimization objective and is the standard choice of barrier on polytopes~\cite{nesterov1994interior}. 
Concretely, starting from  \Cref{eq:reformulated}, the following is our barrier MPC. 

\begin{problem}[Barrier MPC]\label[prob]{def:barr_mpc_formal} Given an MPC as in \Cref{eq:reformulated} and  weight $\eta > 0$, the barrier MPC is defined by minimizing, over the input sequence $\vu\in\R^{\T\cdot d_u}$, the cost
\[ 
\begin{array}{ll}
\vspace{0.5em}
\calV^\eta(\vx_0, \vu)&:= \tfrac{1}{2}\vu^\top \mH \vu  - \vx_0^\top \mF \vu - \eta \left[\mathbf{1}^\top\log(\phi(\vx_0, \vu)) - \vd^\top \vu \right],
\end{array}\numberthis\label{eq:barrierMPC}\] where $\phi(\vx_0, \vu)= \mP \vx_0 + \vw - \mG \vu\in \R^{\nconstr}$
is the (vector) residual of constraints for $\vx_0$ and $\vu$, and the vector $\vd$ is set to  $\vd :=\nabla_{\vu} \sum_{i=1}^m\log(\phi_i(0, \vu)) \vert_{\vu=0}$. We denote by  $\vueta(\vx_0)$  the minimizer of \Cref{eq:barrierMPC} for a given $\vx_0$ and by   $\pibmpc(\vx) := \argmin_{\vu_0} \min_{\vu_{1:\T-1}} \calV^\eta(\vx, \vu)$ the associated control policy.
\end{problem}
Some remarks are in order. First, the choice of $\vd$ in \Cref{def:barr_mpc_formal} is made so as to ensure that $\argmin_{\vueta} \calV^\eta(0, \vueta) = 0$, i.e. that $\pibmpc$ satisfies $\pibmpc(0) = \pimpc(0) = 0$, which is a necessary condition for the controller to be stabilizing at the origin. Further, note that $\|\vd\|^2$ is a constant by construction, a fact that turns out to be useful in \Cref{thm:hess_ueta_bounded}. 
Secondly, the technical assumptions about the constraint polytope in \Cref{eq:reformulated} containing a full-dimensional ball of radius $r$ and being contained inside an origin-centered ball of radius $R$ are both inherited by \Cref{def:barr_mpc_formal}. 

\subsection{Error Bound for Barrier MPC}\label{sec:error_bound_barrier_mpc}
To kick off our analysis of the barrier MPC, we first give the following upper bound on the distance between  the optimal solution of \Cref{def:barr_mpc_formal} and that of explicit MPC in \Cref{eq:reformulated}. Our result is based on standard techniques to analyze the sub-optimality gap in interior-point methods and crucially uses the strong convexity of our quadratic cost  in \Cref{eq:barrierMPC}.   
\begin{theorem}\label{thm:error_bound_barrier_mpc}
 Suppose that $\vueta$ and $\vu^\star$ are, respectively, the optimizers of \Cref{def:barr_mpc_formal} and \Cref{eq:reformulated}. Then we have the following bound in terms of the barrier parameter $\eta$ in  \Cref{eq:barrierMPC}:
\[  \|\vueta - \vu^\star\|\leq O(\sqrt{\eta}). \]
\end{theorem}
\begin{proof} In this proof, we use $\calK$ for the constraint polytope of \Cref{eq:reformulated}. First,  \Cref{lem:linear_plus_barrier_sc}  shows that the recentered log-barrier $\phiK$ in \Cref{def:barr_mpc_formal} is also a self-concordant barrier with some self-concordance parameter $\nu$. 
Since $\vueta= \arg\min_\vu q(\vu) + \eta \phiK(\vu)$, where $q$ is the quadratic cost function of \Cref{def:barr_mpc_formal} and $\phiK$ the recentered log-barrier on $\calK$, we have by first-order optimality:  \[ \nabla q(\vueta) = - \eta \nabla\phiK(\vueta). \numberthis\label{eq:opt_ueta_err_bound}\] Denote by $\alpha$ the strong convexity parameter of the cost function in  \Cref{def:barr_mpc_formal} and by $\nu$ the self-concordance parameter of the barrier $\phiK$. 
Then,  
\begin{align*}
    \left\{ q(\vueta) - q(\vu^\star) \right\} + \frac{1}{2}\alpha \|\vueta-\vu^\star\|^2 \leq \nabla q(\vueta)^\top (\vueta - \vu^\star)
    = \eta \cdot\nabla \phiK(\vueta)^\top (\vu^\star - \vueta) 
    \leq \eta \nu, 
\end{align*} where the first step is by $\alpha$-strong convexity of $q$,  the second step uses \Cref{eq:opt_ueta_err_bound}, and the final step applies \Cref{thm:inner_prod_ub_nu} at the points $\vueta$ and $\vu^\star$. 
Since both $q(\vueta) - q(\vu^\star)$ and $\frac{1}{2}\alpha\|\vueta-\vu^\star\|^2$ are positive, we can bound the latter   by $\eta \nu$. Finally, note that $\nu\geq1$ to finish the proof. 

\end{proof}

We note that the above bound of $\mathcal{O}(\sqrt{\eta})$ is tight for arbitrary directions. However, provided that $\vu^\star \neq \mK_0 \vx_0$ (where $\mK_0$ is the gain associated with the origin piece of the explicit MPC, i.e. the solution is not in the interior of the constraint set), we can show that there exists a direction (independent of $\eta$ of the choice of barrier) along which the error scales with $O(\eta/\|\vu^\star - \mK_0\vx_0\|)$ 
for small $\eta$.

\begin{lemma}\label{thm:error_bound_barrier_directional}Suppose that $\vueta$ and $\vu^\star$ are, respectively, the optimizers of \Cref{def:barr_mpc_formal} and \Cref{eq:reformulated}. Consider the case where $\vu^\star \neq \mK_0 \vx_0$, for $\mK_0 = \mH^{-1}\mF^\top$, i.e. $\mK_0\vx_0$ is the solution to the unconstrained problem. 
 Let $\va = \mH(\vu^\star - \mK_0 \vx_0)/\|\mH(\vu^\star - \mK_0 \vx_0)\|$. Then we have the following upper and lower bounds in terms of the barrier parameter $\eta$ in  \Cref{def:barr_mpc_formal}:
\[\va^\top(\vu^\eta - \vu^\star) \leq \frac{1}{2\sqrt{\alpha_1}}\left(\sqrt{{4m \eta} + \|\vu^\star - \mK_0\vx_0\|_{\mH}^2} - \|\vu^\star - \mK_0\vx_0\|_{\mH}\right),\]
\[\sqrt{\frac{\alpha_1}{\alpha_2}}\cdot\frac{r}{R} \min\left\{\frac{1}{\sqrt{m \alpha_2}}\left(\sqrt{{\eta} + \|\vu^\star - \mK_0\vx_0\|_{\mH}^2} - \|\vu^\star -\mK_0\vx_0\|_{\mH}\right), \sqrt{\frac{\alpha_1}{\alpha_2}}\cdot\frac{r}{2m + 4 \sqrt{m}}\right\} \leq \va^\top(\vu^\eta - \vu^\star),\]
where $m$ is the number of constraints and $\alpha_1\mI \preceq \mH \preceq \alpha_2\mI$.
\end{lemma}
\begin{proof}This is a direct application of \Cref{thm:quad_opt_result}, which holds for general $\nu$-self-concordant barriers, using $\nu = m$, the self-concordant barrier parameter of our log barrier.
\end{proof}
Since by \Cref{thm:best_smooth_bound}, the spectral norm of the Hessian of $\vueta$ is lower-bounded by $\mathcal{O}(1/\epsilon)$ for $\mathcal{O}(\epsilon)$ error, the upper and lower error bounds of \Cref{thm:error_bound_barrier_directional} suggest the tightest-possible upper bound on the Hessian that could be shown (and which would demonstrate barrier MPC is an optimal smoother along the $\va$ direction) is $\mathcal{O}\left(1/\left[\sqrt{\eta + \|\vu^\star - \mK_0 \vx_0\|^2} - \|\vu^\star - \mK_0 \vx_0\|\right]\right).$ We indeed establish this bound  later in \Cref{thm:hess_ueta_bounded}.

The above bound highlights that barrier MPC smooths in a manner which is shaped by problem constraints, unlike randomized smoothing, which generally smooths isotropically. Namely, note that the direction $\va$ is the direction of the gradient of the objective at $\vu^\star$, which is a combination active constraint directions. This is indicative of barrier MPC smoothing less aggressively along directions associated with active constraints.

\subsection{First-Derivative Bound for the Barrier MPC}\label{sec:analysis_barrier_MPC}
To prove our main result (\Cref{thm:hess_ueta_bounded}) on the spectral norm of the Hessian, we first  establish the following technical result bounding the first derivative of $\vueta$ with respect to $\vx_0$. This result may be of independent interest, since it formulates the Jacobian of the log-barrier smoothed solution as a convex combination of derivatives associated with sets of active constraints from the original MPC problem.
Our proof starts with the first-order optimality condition for $\vueta$ and obtains the desired simplification by applying the Sherman-Morrison-Woodbury identity (\Cref{fact:shermanMorrisonWoodbury}). 

\begin{lemma}\label{lem:dueta_dx}
Consider \Cref{def:barr_mpc_formal} with associated cost  matrices $\mH$ and $\mF$ defined therein. Let $\Phi := \Diag(\phi(\vx_0, \vueta(x_0)))$ be the diagonal matrix constructed using the (vector) residual $\phi(x_0, \vueta(x_0)) = \mP \vx_0 + \vw - \mG \vu\in \R^{\nconstr}$. Then, the solution $\vueta$ to the barrier  MPC  in \Cref{def:barr_mpc_formal} satisfies: 
\begin{align*}\label{eq:DuetaDx}
    \DuetaDx &= \mH^{-1}[\mF^\top - \mG^\top(\mG \mH^{-1}\mG^\top + \eta^{-1} \Phi^{2})^{-1}(\mG \mH^{-1} \mF^\top - \mP)].
\end{align*} 
\end{lemma}
\begin{proof}
 We first state the following first-order optimality condition associated with minimizing  \Cref{eq:barrierMPC}:
\[ \mH \vueta(\vx_0)- \mF^\top \vx_0 + \eta \sum_{i = 1}^m \left(\frac{\vg_i}{\phi_i(\vx_0, \vueta(\vx_0))} + \vd_i\right) = 0.\]
Differentiating with respect to $\vx_0$ and rearranging yields \[\DuetaDx = (\mH + \eta \mG^\top \Phi^{-2} \mG)^{-1} (\mF^\top + \eta \mG^\top \Phi^{-2} \mP).\numberthis\label{eq:duetadx_first_eq}\]
For the rest of the proof, we introduce the notation $\mS=\mG\mH^{-1}\mG^{\top}+\eta^{-1}\Phi^{2}$.
Then, we have by applying the Sherman-Morrison-Woodbury identity  (\Cref{fact:shermanMorrisonWoodbury}) to the inverse in \Cref{eq:duetadx_first_eq} that 
\[
(\mH+\eta\mG^{\top}\Phi^{-2}\mG)^{-1}=\mH^{-1}-\mH^{-1}\mG^{\top}\mS^{-1}\mG\mH^{-1},
\]
which simplifies our expression in \Cref{eq:duetadx_first_eq} to 
\begin{align*}
\DuetaDx & =(\mH^{-1}-\mH^{-1}\mG^{\top}\mS^{-1}\mG\mH^{-1})\cdot(\mF^{\top}+\eta\mG^{\top}\Phi^{-2}\mP)\\
 & =\mH^{-1}\mF^{\top}-\mH^{-1}\mG^{\top}\mS^{-1}\mG\mH^{-1}\mF^{\top}+\underbrace{(\mH^{-1}-\mH^{-1}\mG^{\top}\mS^{-1}\mG\mH^{-1})\cdot\eta\mG^{\top}\Phi^{-2}\mP}_{\text{Term 1}}. \numberthis\label{eq:duetadx_second_eq}
\end{align*}
We now show that ``Term 1'' may be simplified as follows. 
\begin{equation}
(\mH^{-1}-\mH^{-1}\mG^{\top}\mS^{-1}\mG\mH^{-1})\cdot\eta\mG^{\top}\Phi^{-2}\mP=\mH^{-1}\mG^{\top}\mS^{-1}\mP.\label{eq:duetadx_third_eq}
\end{equation}
Once this is done, the claim is finished, since  plugging the right-hand side  from \Cref{eq:duetadx_third_eq} into  ``Term 1'' from \Cref{eq:duetadx_second_eq} gives exactly the claimed expression in the statement of the lemma.  
Therefore, we now prove \Cref{eq:duetadx_third_eq}. To this end, we observe that by  factoring out
$\mH^{-1}\mG^{\top}$ from the left and $\eta \Phi^{-2}\mP$ from the right, we
may re-write the left-hand side in \Cref{eq:duetadx_third_eq} as 
\begin{align*}
(\mH^{-1}-\mH^{-1}\mG^{\top}\mS^{-1}\mG\mH^{-1})\cdot\eta\mG^{\top}\Phi^{-2}\mP  
 &\,=\mH^{-1}\mG^{\top}({I}-\mS^{-1}\mG\mH^{-1}\mG^{\top})\cdot\eta\Phi^{-2}\mP\\
 &\,=\mH^{-1}\mG^{\top}\mS^{-1}\cdot(\mS-\mG\mH^{-1}\mG^{\top})\cdot\eta\Phi^{-2}\cdot\mP\\
 &\,=\mH^{-1}\mG^{\top}\mS^{-1}\mP,
\end{align*}
where the last step is by using our definition of $\mS$  and cancelling $\eta^{-1}\Phi^2$ with $\eta\Phi^{-2}$.
\end{proof} 

Equipped with \Cref{lem:dueta_dx}, we are now ready to state \Cref{thm:convex_combination}, where we connect the rate of change (with respect to the initial state $x_0$) of the solution \Cref{eq:K_sigma} of the constrained MPC and the  barrier MPC solution (from \Cref{lem:dueta_dx}). Put simply, \Cref{thm:convex_combination} tells us that {the  barrier MPC solution implicitly interpolates between a potentially exponential number of affine pieces from the original explicit MPC problem}. This important connection helps us get a handle on the smoothness of barrier MPC in \cref{thm:hess_ueta_bounded}. The starting point for our proof for this result is the expression for $\DuetaDx$ from \Cref{lem:dueta_dx}. To simplify this expression so it is $\eta$-independent, we  crucially use our linear algebraic result (\Cref{lem:svd_adj}) on products of the form $\mL \cdot \adj(\mL\mL^\top)_{\bsigma}$ for which $\det(\mL\mL^\top)_{\bsigma}=0$.  

\begin{lemma}\label{thm:convex_combination} 
Consider the setup in \Cref{def:barr_mpc_formal} with associated cost matrices $\mH$ and $\mF$, constraint matrices $\mP$ and $\mG$, and barrier parameter $\eta$, all defined therein. We define the following quantities. 
\begin{enumerate}[label=(\roman*)]
\compresslist{
\item For any $\bsigma\in \left\{0, 1\right\}^m$, define the matrix $ K_{\bsigma} =  \mH^{-1}[\mF^\top - \mG^\top(\mG \mH^{-1} \mG^\top)_{\bsigma}^{-1}(\mG \mH^{-1} \mF^\top - \mP)],$ which, recall, in \Cref{thm:dudx} describes the solution $\vu$ to the constrained MPC. 
\item Recall from \Cref{def:barr_mpc_formal} the residual $\phi:= \phi(\vx_0, \vu)= \mP\vx_0 + \vw - \mG\vu\in \R^{\nconstr}$; here, we denote it by $\phi$. For any   $\phi$ and  $\bsigma=\left\{0, 1\right\}^m$, define the scaling factor $h_{\bsigma} = \det([\mG \mH^{-1}\mG^\top]_{\bsigma})\prod_{i=1}^m (\eta^{-1}\phi^2_i)^{1 - \sigma_i}.$
\item We split the set $\bsigma\in \left\{0, 1\right\}^m$ into the following two sets: \[ S := \left\{\bsigma \in \left\{0, 1\right\}^{\nconstr} \, \mid \, \det([\mG\mH^{-1}\mG^\top]_{\bsigma}) > 0\right\} \text{ and } S^{\complement} = \left\{\bsigma\in \left\{0, 1\right\}^m \, \mid \, \det([\mG\mH^{-1}\mG^\top]_{\bsigma})=0 \right\}.\] 
}
\end{enumerate}
Then the rate of evolution, with respect to $x_0$, of the solution $\vueta$ to the  barrier MPC (in \Cref{lem:dueta_dx}) 
is connected to the solution of the constrained MPC (in \Cref{eq:K_sigma}) as follows:
\begin{align*}
\DuetaDx = \frac{1}{\sum_{\bsigma \in S}h_{\bsigma}}\sum_{\bsigma \in S} h_{\bsigma} \mK_{\bsigma}.
\end{align*}
\end{lemma}
\begin{proof}Let $\Phi := \Diag(\phi)$. Then from \Cref{lem:dueta_dx}, we have the following expression for $\DuetaDx$: 
\begin{align*}
    \DuetaDx &= \mH^{-1}[\mF^\top - \mG^\top(\mG \mH^{-1}\mG^\top + \eta^{-1} \Phi^{2})^{-1}(\mG \mH^{-1} \mF^\top - \mP)].
\end{align*} We now split $\mG^\top(\mG\mH^{-1}\mG^\top + \eta^{-1}\Phi^2)$ above into the following two components via \Cref{lem:split_into_adj}. 
\[ \mG^\top (\mG\mH^{-1}\mG^\top + \eta^{-1} \Phi^2)^{-1} = \frac{1}{\sum_{\bsigma\in S} h_{\bsigma}} \left(\sum_{\bsigma\in S}{h_{\bsigma}}\cdot\mG^\top (\mG\mH^{-1}\mG^\top)^{-1}_{\bsigma} + \sum_{\bsigma\in S^{\complement} }c_{\bsigma}\cdot\mG^\top \adj(\mG\mH^{-1}\mG^\top)_{\bsigma}\right), \] where
 $c_{\bsigma} := \prod_{i=1}^m (\eta^{-1}\phi_i^2)^{1-\sigma_i}$.  By definition of $S^{\complement}$, the second sum in the preceding equation comprises those terms for which $\det(\mG\mH^{-1}\mG^\top)_{\bsigma}=0$. We now invoke \Cref{lem:Gtop_adjGHinvGtopsigma_zero}, which states that $\mG^\top \adj(\mG\mH^{-1}\mG^\top)_{\bsigma}=\zero$ for all $\bsigma\in S^{\complement}$. Consequently, we may express $\DuetaDx$ in terms of only the first set of terms in the preceding equation (zeroing out the second set of terms): 
\begin{align*}
    \DuetaDx
    &= \frac{1}{\sum_{\bsigma} h_{\bsigma}}\sum_{\bsigma \in S} h_{\bsigma} \mH^{-1}[\mF^\top - \mG^\top(\mG \mH^{-1}\mG^\top)_{\bsigma}^{-1}(\mG \mH^{-1} \mF^\top - \mP)] 
     = \frac{1}{\sum_{\bsigma} h_{\bsigma}} \sum_{\bsigma \in S} {h_{\bsigma}} \mK_{\bsigma}.
\end{align*}  where we plugged in \Cref{eq:K_sigma} in the final step, thus concluding the proof.  
\end{proof} The above theorem immediately implies that  $\left\|\DuetaDx\right\|$ is bounded from above as stated next in \cref{cor:first_der_bound_ueta}. We draw attention to the fact that \cref{cor:first_der_bound_ueta} shows that the Lipschitz constant of $\vueta$ is independent of $\eta$, which  demonstrates that the log-barrier does not worsen the Lipschitz constant of the controller, rather it  changes only the interpolation between the different pieces.
\begin{corollary}\label{cor:first_der_bound_ueta}
In the setting of \Cref{thm:convex_combination}, we have,
\begin{align*}
    \left\|\DuetaDx\right\| \leq L := \max_{\sigma \in S} \|K_\sigma\|. 
\end{align*}
\end{corollary}
\begin{proof}
    From \Cref{thm:convex_combination}, we  infer that $\frac{\partial \vueta}{\partial x_0}$  lies in the convex hull of $\{K_\sigma\}_{\sigma \in S}$, and note that $|S| < \infty$.
\end{proof}

\subsection{Main Result: Smoothness Bound for the Barrier MPC} \label{sec:smoothness_bound_barrier_mpc}
We are now ready to state our main result, which effectively shows that $\vueta$ (and hence $\pibmpc$) satisfies the conditions of \Cref{assumption:smoothness_of_exp_and_learned}. Our proof of \Cref{thm:hess_ueta_bounded} starts with \Cref{lem:dueta_dx} and computes another derivative. To get an upper bound on the operator norm of the Hessian so obtained, our proof then crucially hinges on  \Cref{lem:contains_linear_ball} and \Cref{thm:quad_opt_result}, which provide explicit lower bounds on residuals when minimizing a quadratic cost plus a self-concordant barrier over a polytope, a result we believe to be of independent interest to the optimization community. 

\begin{theorem}\label{thm:hess_ueta_bounded} 
Consider the setting of \Cref{def:barr_mpc_formal} with associated cost matrices $\mH$ and $\mF$, constraint matrices $\mP$ and $\mG$, barrier parameter $\eta$, number of constraints $\nconstr$, the recentering vector $\vd$, and the solution $\vueta$, all defined therein. We define the following quantities. 
\begin{enumerate}[label=(\roman*)]
\compresslist{
\item Denote by $L$ the Lipschitz constant of $\vueta$ from 
\Cref{cor:first_der_bound_ueta}. 
\item We split the set $\bsigma\in \left\{0, 1\right\}^m$ into the following two sets: \[ S := \left\{\bsigma \in \left\{0, 1\right\}^{\nconstr} \, \mid \, \det([\mG\mH^{-1}\mG^\top]_{\bsigma}) > 0\right\} \text{ and } S^{\complement} = \left\{\bsigma\in \left\{0, 1\right\}^m \, \mid \, \det([\mG\mH^{-1}\mG^\top]_{\bsigma})=0 \right\}.\] 
\item For $\bsigma\in S$, define the parameter   $C := \max_{\bsigma \in S} \|2\mH^{-1}\mG^\top (\mG \mH^{-1} \mG^\top)_{\bsigma}^\dagger\|$. 
\item Denote by $r$ and $R$ the inner and outer radius, respectively, associated with  \Cref{def:barr_mpc_formal}.   
\item Define the residual lower bound, 
\[\philb =\frac{\lambda_{\min}(H)}{\lambda_{\max}(H)}\cdot\frac{r}{R}\cdot \min\left\{\frac{1}{\sqrt{\nu {\lambda_{\min}(\mH)}}}\left(\sqrt{{\eta} + \|\vu^\star - \mK_0\vx_0\|_{\mH}^2} - \|\vu^\star -\mK_0\vx_0\|_{\mH}\right), \frac{r}{2\nu + 4 \sqrt{\nu}}\right\},\numberthis\label{eq:philb}\] 
with  $\nu=20(\nconstr+R^2\|\vd\|^2)$ and $\mK_0 \vx_0 := \mH^{-1}\mF^\top \vx_0$, the solution to the unconstrained minimization of the quadratic objective. We denote $\|\vu\|_{\mH} := \sqrt{\vu^\top\mH\vu}$.
}
\end{enumerate}
Then, the Hessian of $\vueta$ with respect to $\vx_0$ is bounded by:
\begin{align*}
    \left\|\frac{\partial^2 \vueta}{\partial \vx_0^2}\right\| \leq \frac{C}{\philb}(\|\mP\| + \|\mG\|L)^2.
\end{align*} 
where $\|\cdot\|$ denotes the spectral norm of the  third-order tensor $\frac{\partial^2 \vueta}{\partial \vx_0^2}$.
\end{theorem}
\begin{proof} Recall from \Cref{lem:dueta_dx} the following expression for $\DuetaDx$ evaluated at a particular $\vx_0$:
\begin{align*}
    \DuetaDx(\vx_0) &= \mH^{-1}[\mF^\top - \mG^\top(\mG \mH^{-1}\mG^\top + \eta^{-1} \Phi(\vx_0)^{2})^{-1}(\mG \mH^{-1} \mF^\top - \mP)],
\end{align*}
where $\Phi(\vx_0) := \Diag(\mP\vx_0 - \vw + \mG\vueta(\vx_0))$. Let $\vy \in \R^{d_x}$ be an arbitrary unit-norm vector, and define the univariate function \[\mM(t) := \mG\mH^{-1}\mG^\top + \eta^{-1}\Phi(t)^2,\] where 
we overload $\Phi$ to mean $\Phi(t) := \Diag(\mP(\vx_0 + t\vy) - \vw + \mG\vueta(\vx_0 + t\vy))$, the residual along the path $t\mapsto \vx_0 + t\vy$. We therefore have the following expression for $\DuetaDx$ evaluated at $\vx_0 + t\vy$:
\begin{align*}
    \DuetaDx(\vx_0 + t\vy)  &= \mH^{-1}[\mF^\top - \mG^\top\mM(t)^{-1}(\mG \mH^{-1} \mF^\top - \mP)].
\end{align*} Then by differentiating $\mM(t)^{-1}$ and applying the chain rule, we get,
\begin{align*}
    \frac{d}{dt}\left(\frac{\partial \vueta}{\partial \vx_0}(\vx_0 + t \vy)\right) 
    = &\,\mH^{-1}\mG^\top \mM(t)^{-1}\frac{d\mM(t)}{dt} \mM(t)^{-1}(\mG\mH^{-1} \mF^\top - \mP) \\
    = &\,2\mH^{-1}\mG^\top \mM(t)^{-1}\left(\frac{d\Phi(t)}{dt} \eta^{-1} \Phi(t) \right)\mM(t)^{-1}(\mG\mH^{-1} \mF^\top - \mP) \\
    = &\, 2\mH^{-1}\mG^\top \mM(t)^{-1}  \frac{d\Phi(t)}{dt} (\eta \mG \mH^{-1}\mG^\top\Phi^{-1}(t) + \Phi(t))^{-1}(\mG\mH^{-1} \mF^\top - \mP) \\
    = &\, 2\mH^{-1}\mG^\top \mM(t)^{-1}  \frac{d\Phi(t)}{dt} (\eta \Phi(t)^{-1}\mG \mH^{-1}\mG^\top\Phi(t)^{-1} + \mI)^{-1}\Phi(t)^{-1}(\mG\mH^{-1} \mF^\top - \mP),
\end{align*} where the third and fourth steps factor out $\Phi(t)$ from the right and left, respectively. We now bound groups of terms of the product on the right-hand side and then finish the bound by submultiplicativity of the spectral norm. First,   since $\mM(t)$ is a sum of a square matrix and a positive diagonal matrix, we may 
apply \Cref{lem:split_into_adj} to express $\mG^\top \mM(t)^{-1}$ as follows with appropriate $h_{\bsigma}$ and $c_{\bsigma}$:
\[  \mG^\top\mM(t)^{-1} = \frac{1}{\sum_{\bsigma\in S} h_{\bsigma}} \left(\sum_{\bsigma\in S}{h_{\bsigma}}\mG^\top (\mG\mH^{-1}\mG^\top)_{\bsigma}^\dagger + \sum_{\bsigma\in S^{\complement}} {c_{\bsigma}} \mG^\top\adj(\mG\mH^{-1}\mG^\top)_{\bsigma}\right).\numberthis\label{eq:gtop_intermediate}\] Now note that for $\bsigma \in S^{\complement}$, we have $\det(\mG\mH^{-1}\mG^\top)_{\bsigma}=0$. 
We may then invoke  \Cref{lem:Gtop_adjGHinvGtopsigma_zero} to infer that 
 for $\bsigma \in S^{\complement}$, we have $\mG^\top \adj(\mG \mH^{-1}\mG^\top)_{\bsigma} = \zero$.  
As a result, the second term on the right-hand side of \Cref{eq:gtop_intermediate} vanishes, 
thereby affording us the following simplification: 
\begin{align*}
    \|2\mH^{-1}\mG^\top \mM(t)^{-1}\|
    &= \left\| \frac{1}{\sum_{\bsigma \in S} h_{\bsigma}}\sum_{\bsigma \in S} 2h_{\bsigma}\mH^{-1}\mG^\top(\mG \mH^{-1}\mG^\top)_{\bsigma}^\dagger \right\| \\
    &\leq C := \max_{\bsigma \in S} \|2\mH^{-1}\mG^\top (\mG \mH^{-1} \mG^\top)_{\bsigma}^\dagger\|,
\end{align*}
where the second equality follows via H\"older's inequality.
Next, from the definition of $\Phi$, we have that $\frac{d\Phi}{dt} = \mP + \mG \left(\DuetaDx(\vx_0 + t\vy)\vy\right)$. By the triangle inequality,  the Lipschitzness $L$ of $\vueta$ (from \Cref{cor:first_der_bound_ueta}), and the fact that $\vy$ is unit norm, we have  
\[\left\|\frac{d\Phi}{d t}\right\| \leq \|\mP\| + \|\mG\| \left\|\DuetaDx\right\| \leq \|\mP\| + \|\mG\| L.\] 
To bound $\|(\eta \Phi^{-1}\mG\mH^{-1} \mG^\top \Phi^{-1} + \mI)^{-1}\Phi^{-1}\|$, we first note that because $\eta \Phi^{-1}\mG\mH^{-1} \mG^\top \Phi^{-1} \succeq \zero$, we have $(\eta \Phi^{-1}\mG\mH^{-1} \mG^\top \Phi^{-1} + \mI)^{-1}\preceq \mI$, which in turn implies that $\|(\eta \Phi^{-1}\mG\mH^{-1} \mG^\top \Phi^{-1} + \mI)^{-1}\|\leq 1$. Then, by submultiplicativity of the spectral norm, we have 
$$\|(\eta \Phi^{-1}\mG\mH^{-1} \mG^\top \Phi^{-1} + \mI)^{-1}\Phi^{-1}\| \leq \|\Phi^{-1}\| \leq \frac{1}{\min_{i \in [\nconstr]} \phi_i}.$$
We may then plug in the lower bound on $\min_{i \in [\nconstr]} \phi_i$ from \Cref{thm:quad_opt_result} that uses $\nu = 20(m+R^2\|d\|^2)$, the self-concordance parameter (computed via \Cref{lem:linear_plus_barrier_sc}) of the recentered log-barrier in \Cref{def:barr_mpc_formal}. 
Finally, recognizing $\mH^{-1} \mF^\top$ as $K_{\bsigma}$ from \Cref{eq:K_sigma} (with $\bsigma = \zero^\nconstr$) yields \[ \|\mG \mH^{-1}\mF^\top - \mP\| = \left\|\mG K_{0} - \mP\right\| \leq \|\mP\| + \|\mG\|L.\] 
Combining all the bounds obtained above, we may then finish the proof. 
\end{proof}
Thus, \Cref{thm:hess_ueta_bounded} establishes bounds analogous to those  in \Cref{thm:randomized_bounds} for randomized smoothing, demonstrating that the Jacobian of the smoothed expert policy is sufficiently Lipschitz. Indeed, in this case our result is stronger, showing that the Jacobian is differentiable and  the Hessian tensor is bounded. This theoretically validates the  core proposition of our paper:  the barrier MPC policy in \Cref{def:barr_mpc_formal} is suitably smooth, and therefore the  guarantees in  \Cref{sec:learning_guarantees} hold. 
We now briefly revisit our learning guarantees, applied to specifically to log-barrier MPC.

\subsection{Learning Guarantees for Barrier MPC}
\label{sec:end_to_end}
We now revisit the learning guarantees discussed in \Cref{sec:learning_guarantees}, adapted specifically to a log-barrier MPC expert.
We begin by considering the stability properties of barrier MPC. Since we are  interested in  establishing $\|\hat{\vx}_t - \vx^\star_t\| \leq \epsilon$, where $\hat{\vx}$ is the state under the learned policy and $\vx^\star$ is the state under the expert, and since we consider MPC controllers which stabilize to the origin, we can relax our local incremental input-to-state stability requirements to simply input-to-state stability (ISS) with minimal assumptions. \Cref{def:inp_stab} introduces this weaker input-to-state stability property, and \Cref{lem:iss_diss} shows that ISS policies are locally $\delta$ISS. We then observe that there is considerable prior work showing that ISS holds under minimal assumptions for barrier MPC, meaning \Cref{assum:stable} is satisfied for barrier MPC.
    
\begin{definition}[Input to State Stability \cite{khalil2002control}]\label{def:inp_stab} A system $\vx_{t+1} = f(\vx_t, \vu_t)$ is input-to-state stable (equivalently, a controller $\pi$ is input-to-state stabilizing under $f$ for $\vx_{t+1} = f(\vx_t, \pi(\vx_t) + \vu_t)$) if there exists $\beta \in \mathcal{KL}$ and $\gamma \in \mathcal{K}$ (where $f \in \mathcal{K}$ provided $f: [0,a] \to \R_{\geq 0}$ for $a > 0$, $f$ is strictly increasing, $f(0) = 0$, and $g \in \mathcal{KL}$ provided $g: [0,a] \times \R_{\geq 0} \to \R_{\geq 0}$ is class $\mathcal{K}$ in the first argument for some $a > 0$ and monotonically decreasing in the second such that $\lim_{t \to \infty} g(s,t) = 0 \forall s$).  so that for all initial values $\vx_0$ and all admissible inputs $\vu_t$, we have 
\begin{align*}
    \|\vx_t\| \leq \beta(\|\vx_0\|,t) + \gamma(\|\vu_t\|_\infty) \quad \forall t \geq 0.
\end{align*}
\end{definition}

\begin{lemma}\label{lem:iss_diss}Let $\pi$ be an $L$-Lipschitz controller which is input-to-state stabilizing for the dynamics $\vx_{t+1} = \mA \vx_t + \mB \vu_t$ with  gains $\beta \in \mathcal{KL}, \gamma \in \mathcal{K}$. Define $\mathcal{B}^{-1}$ such that $\beta(\|\vx_0\|, \mathcal{B}^{-1}(\epsilon)) \leq \epsilon$ for $\|\vx_0\| \leq B_x$. As $\beta \in \mathcal{KL}$, we know that $\mathcal{B}^{-1}$ exists and is monotonically decreasing. Define the gain,
\[v(\epsilon) := \min\left\{\gamma^{-1}(\epsilon/2), \epsilon \cdot (1 + \|\mA\| + (1 + L)\|\mB\|)^{-\mathcal{B}^{-1}(\epsilon/4)}\right\}.\]
Then, $\pi$ is incrementally input-to-state stabilizing with gain $\gamma' := v^{-1} \in \mathcal{K}$, i.e, for $\|\vu\|_{\infty} \leq v(\epsilon)$ we have that $\|\vx_t - \bar{\vx}_t\| \leq \epsilon$ where $\vx_{t+1} = f(\vx_{t}, \pi(\vx_t) + \vu_t)$ and $\bar{\vx}_{t+1} = f(\bar{\vx}_{t}, \pi(\vx_t))$ with $\bar{\vx}_0 = \vx_0$. Note that over a horizon of length $K$, we have
\begin{align*}
    \gamma'(\|\vu\|_{\infty}) \leq \max\{ 2\gamma(\|\vu\|_{\infty}), (1 + \|\mA\| + (1 + L)\|\mB\|)^{K} \|\vu\|_{\infty} \}.
\end{align*}
\end{lemma}

\begin{proof}
We first show that $\mathcal{\gamma}' \in \mathcal{K}$. Since $\gamma \in \mathcal{K}, \gamma^{-1} \in \mathcal{K}$. Furthermore, as $\mathcal{B}^{-1}$ is monotonically decreasing, $C^{-\mathcal{B}^{-1}(\epsilon)}$ is monotonically non-decreasing in $\epsilon$ for $C \geq 1$ and $\epsilon \cdot C^{-\mathcal{B}^{-1}(\epsilon)} \in \mathcal{K}$. Since $v$ is the minimum of two class $\mathcal{K}$ gains, it follows that $v \in \mathcal{K}$ and therefore $\gamma' := v^{-1} \in \mathcal{K}$.

We now prove that $\pi$ is incrementally input-to-state stabilizing with gain $\gamma' := v^{-1}$. Fix any $\epsilon > 0$. WTS that for $\|\vu\|_{\infty} \leq v(\epsilon)$, $\|\vx - \bar{\vx}\|_{\infty} \leq \epsilon$. First, consider $t \leq \beta^{-1}(\epsilon/4)$. 
Note that,
\begin{align*}
    \|\vx_{t+1} - \bar{\vx}_{t+1}\| \leq (\|\mA\| + \|\mB\|L)\|\vx_t - \bar{\vx}_t\| + \|\mB\| \cdot \|\vu\|_\infty.
\end{align*}
By telescoping we can write,
\[\|\vx_{t} - \bar{\vx}_{t}\| \leq (1 + \|\mA\| + (1+L)\|\mB\|)^t\cdot\|\vu\|_{\infty} \quad.\numberthis\label{eq:error_bound_a}\]
We then use that $t \leq \mathcal{B}^{-1}(\epsilon/4)$ and that $\|\vu\|_{\infty} \leq v(\epsilon) \leq \epsilon (1 + \|\mA\| + (1 + L)\|\mB\|)^{-\mathcal{B}^{-1}(\epsilon/4)}$. Combining with \Cref{eq:error_bound_a}, we get $\|\vx_t - \bar{\vx}_t\| \leq \epsilon$.
We next consider the case $t \geq \mathcal{B}^{-1}(\epsilon/4)$. We now finish the proof by using the input-to-state stability of $\pi$:
\begin{align*}
    \|\vx_t - \bar{\vx}_t\| &\leq \|\vx_t\| + \|\bar{\vx}_t\| \leq 2\beta(\|\vx_0\|,t) + \gamma(\|\vu\|_{\infty})  
    \leq 2\beta(\|\vx_0\|, \mathcal{B}^{-1}(\epsilon/4)) + \gamma(\gamma^{-1}(\epsilon/2)) 
    \leq \frac{\epsilon}{2} + \frac{\epsilon}{2} \leq \epsilon.
\end{align*}
\end{proof}
For stabilizable systems and proper choices of cost function and constraints, the barrier MPC is ISS~\cite{pouilly2020stability, feller2014barrier}. We therefore impose the following assumption on the parameters of the barrier MPC. 
\begin{assumption}\label{assum:bmpc_iss}The parameters of the barrier MPC controller $\pibmpc$ in \Cref{def:barr_mpc_formal} are chosen such that the system is input-to-state stabilizing. Consequently, by \Cref{lem:iss_diss} and  \Cref{cor:first_der_bound_ueta}, 
it is incrementally input-to-state stabilizing over $t \leq K$ for some with linear gain function $\gamma$. 
\end{assumption}
This shows that $\pibmpc$ satisfies the even weaker notion of locally $\delta$ISS as required in  \Cref{assum:stable}.
We now state our end-to-end learning guarantee, an extension of \Cref{prop:goodness_of_learned_policy}. 

\begin{corollary}\label{thm:final}Let $\pibmpc$ be a barrier MPC as in \Cref{def:barr_mpc_formal} that satisfies \Cref{assum:bmpc_iss} such that it is $(\kappa,\gamma)$-locally-$\delta$ISS for linear gain $\gamma$ over a horizon length $K$. Let $L$ be as defined in \Cref{cor:first_der_bound_ueta} and  overload $\gamma$ to denote the constant associated with $\gamma(\cdot)$. Let $m$ be the number of constraints and $r,R$ be the radii associated with the constraint polytope in \Cref{def:barr_mpc_formal}.

Assume that $\|\vx_0\| \leq B_x$ and let $\pilearned$ be chosen such that, for some $\epsilon_0, \epsilon_1 > 0$ and given $N$ sample trajectories of length $K$ under $\pibmpc$ from an initial condition distribution $\mathcal{D}$,
\begin{align*}
    \mathbb{P}_{\vx_0 \sim \mathcal{D}}\bigg[&\sup_{0 \leq k \leq K}\|\pilearned(\vx_k) - \pibmpc(\vx_k)\| \leq \epsilon_0/N \,\,\, \wedge \,\,\,  \sup_{0 \leq k \leq K}\left\|\frac{\partial\pilearned}{\partial \vx}(\vx_k) - \frac{\partial \pibmpc}{\partial \vx}(\vx_k)\right\| \leq \epsilon_1/N \bigg] \geq 1 - \delta. 
\end{align*}
Let $C, L, \philb > 0$ be defined as in \Cref{thm:hess_ueta_bounded} and $\mP, \mG$ be the matrices associated with the original MPC problem, \Cref{eq:reformulated}. Then, provided that $N \geq \mathcal{O}\left(\max\left\{\frac{C\gamma^2}{\philb}(\|\mP\| + \|\mG\|L), \gamma, \frac{\gamma}{\kappa}\right\} \epsilon_0\right)$, and $N \geq 4\gamma\epsilon_1$, it follows that,
\begin{align*}
    \|\hat{\vx}_t - \vx^\eta_t\| \leq \frac{8\gamma\epsilon_0}{N} \quad \forall\, 0 \leq t \leq K.
\end{align*}
Per $\cref{thm:hess_ueta_bounded}$, $\philb$ scales with either $\mathcal{O}(\eta)$ or $\mathcal{O}(\sqrt{\eta})$, depending on the direction of $\vu^* - \vu^\eta$.  
\end{corollary}
\begin{proof}This is an application of \Cref{prop:goodness_of_learned_policy}, combined with our $L_1$ smoothness bound on the Hessian (\Cref{thm:hess_ueta_bounded}), where, 
$L_1 \leq\frac{C}{\philb}(\|\mP\| + \|\mG\|L).$
\end{proof}

In conclusion, our analysis in this section shows that  our key assumption of locally input-to-state stability (\cref{assum:stable}) of the barrier MPC controller $\pibmpc$ is satisfied. We also showed  smoothness  of $\pibmpc$. Per the results of \citet{pfrommer2022tasil}, these two properties together imply the learning guarantees in \cref{{thm:final}}. This concludes our theoretical analysis of imitation learning using the barrier MPC. 
We now corroborate our theory with numerical experiments. 

\begin{figure*}
    \centering
    \includegraphics[width=0.95\linewidth]{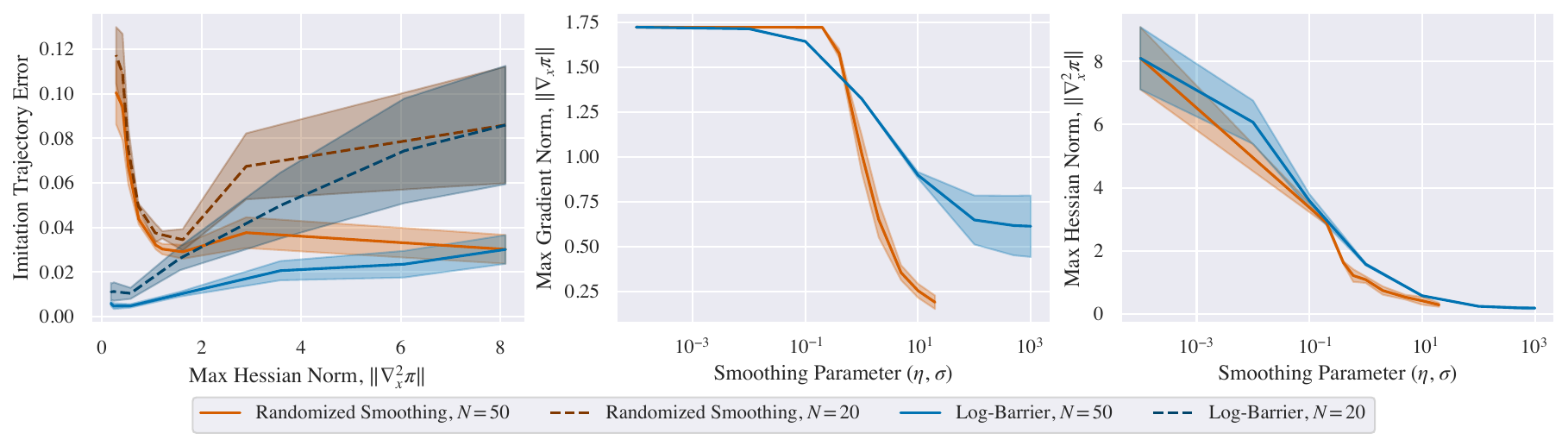}
    \caption{Left: The imitation error $\max_{t} \|\hat{x} - x^\star\|$ for the trained MLP over 5 seeds, as a function of the expert smoothness for both randomized smoothing and log-barrier MPC. Center, Right: The $L_0$ (gradient norm) and $L_1$ (hessian norm) smoothness of $\pi^\star$ as a function of the smoothing parameter. This figure appeared in our previous work \cite{pfrommer2024sample}.}
    \label{fig:experiments}
\end{figure*}

\section{Experiments}
The experiments presented below first appeared in our previous work \cite{pfrommer2024sample}. We include these here for completeness.

We demonstrate the advantage of barrier MPC over randomized smoothing for the double integrator system visualized in \Cref{fig:explicit_mpc}. The matrices describing the dynamics are $\mA = \begin{bsmallmatrix} 1 & 1 \\ 0 & 1\end{bsmallmatrix}$ and $\mB = \begin{bsmallmatrix}0 \\ 1\end{bsmallmatrix}$, 
and the cost matrices are given by $\mQ_t = \mI$, $\mR_t = 0.01 \mI$, with horizon length $T = 10$. Our constraints for \Cref{def:barr_mpc_formal} are $\|\vx\|_{\infty} \leq 10$ and  $\|\vu\|_{\infty} \leq 1$. This is the same setup as in \cite{ahn2023model}, which we note asymptotically stabilizes the system to the origin. 

We sample $N \in \{20, 50\}$ trajectories of length $K = 20$ using $\pibmpc$ and $\pirs$ and smoothing parameters $\eta$ (\Cref{def:barr_mpc_formal}) and $\sigma$ (\Cref{def:rand_smoo}) 
ranging from $10^{-4}$ to $10^3$ and $10^{-4}$ to $20$, respectively. We use $\mathcal{P} = \mathcal{N}(0,I)$ for the randomized smoothing distribution. For each parameter set, we trained a 4-layer multi-layer perceptron (MLP) using GELU activations \cite{hendrycks2016gaussian} to ensure smoothness of $\Pi$. We used AdamW~\cite{loshchilov2018decoupled} with a learning rate of $3 \cdot 10^{-4}$ and weight decay of $10^{-3}$ in order to ensure boundedness of the weights (see \cite{pfrommer2022tasil}).

We visualize the smoothness properties of the chosen expert $\piexpert$ of each method (either $\pibmpc$ or $\pirs$) across the choices of $\eta, \sigma$ in \Cref{fig:experiments}. For small Hessian norms (i.e. the large $\eta, \sigma$ regime), barrier MPC has larger gradient norm $\|\nabla \piexpert\|$ than randomized smoothing. 
This shows that $\pibmpc$ prevents oversmoothing in comparison to $\pirs$. While randomized smoothing reduces $\|\nabla^2 \piexpert\|$ by essentially flattening the function, $\pibmpc$ achieves equally smooth functions while still maintaining control of the system. This effect is also  seen in \Cref{fig:smoothing_contours}, where we visualize the barrier MPC controller for different $\eta$ and see that, even for large $\eta$, we successfully stabilize to the origin.

One interesting phenomenon is that the maximum gradient of $\pibmpc$ begins decreasing much earlier than $\pirs$. This is due to the fact that $\pirs$  smooths only locally, meaning that if the smoothing radius is sufficiently small, the gradient will not be affected. Meanwhile, $\pibmpc$  always performs a \emph{global} form of smoothing, so even for small $\eta$, the controller is smoothed everywhere.

In \Cref{fig:experiments}, we also compare the trajectory error when imitating trajectories from $\pirs, \pibmpc$ for equivalent levels of smoothness. We can see that for $N = 20$ and $N=50$, $\pibmpc$ significantly outperforms $\pirs$ across all smoothness levels. This effect is particularly pronounced in the very smooth 
regime, where imitating $\pirs$ proves unstable due to the inherit instability of $(\mA, \mB)$, leading to extremely large imitation errors. Meanwhile, $\pibmpc$ is strictly easier to imitate the more smoothing that is applied. Overall, these experiments confirm our hypothesis that not all smoothing techniques perform equally and that barrier MPC is an effective smoothing technique that outperforms randomized smoothing for the purposes of imitation learning.

\section{Discussion}

We consider two methods for smoothing MPC policies for constrained linear systems: randomized smoothing and barrier MPC. 
While the former is known to have the theoretically optimal  ratio of approximation error to Hessian norm, 
it may not preserve the stability or constraint satisfaction properties of the underlying controller and hence is not always well-suited for controls applications. 
We show that the log-barrier-based MPC yields a smooth control with  optimal error to smoothness ratio along some direction. 
Additionally, it better ensures constraint satisfaction while also retaining the stability properties of the original policy. 
We show how these properties enable theoretical guarantees when learning barrier MPC and demonstrate experimentally its better performance  compared to a randomized smoothing baseline. 

Our key technical contribution towards proving the smoothness  of barrier MPC is a lower bound  on the optimality gap of the analytic center associated with a convex Lipschitz function, which we hope could be of independent interest to the broader optimization community. Extending our results to smoothing nonlinear MPC policies would be a fruitful direction for future work.

\section*{Acknowledgements}
We gratefully acknowledge funding from  ONR N00014-23-1-2299 and a Vannevar Bush Fellowship from the Office of Undersecretary of Defense for Research and Engineering (OUSDR\&E).

\printbibliography

\newpage

\appendix

\newcommand\SmallMatrix[1]{{%
  \tiny\arraycolsep=0.3\arraycolsep\ensuremath{\begin{pmatrix}#1\end{pmatrix}}}}

\appendix
\section{Technical Results from Matrix Analysis}\label{sec:app-matrix-formulas} 
We use the notation introduced in \Cref{sec:notation}. 
Additionally, we use $\ve_i$ to denote the vector with one at the $i^\mathrm{th}$ coordinate and zeroes at the remaining coordinates. 
We first collect the following relevant facts from matrix analysis before proving our technical results. 

\begin{fact}[\cite{horn2012matrix}]\label{def:adjugate_cofactor}
    Given a square matrix $\mA\in \R^{n\times n}$, its $(i, j)^\mathrm{th}$ minor $\mM_{i,j}$, is defined as  the determinant of the $(n-1)\times (n-1)$ matrix resulting from deleting row $i$ and column $j$ of $\mA$. Next, the $(i,j)^\mathrm{th}$ cofactor is defined to be the $(i,j)^\mathrm{th}$ minor scaled by $(-1)^{i+j}$: \[\mC_{ij} = (-1)^{i+j}\mM_{i,j}.\numberthis\label{eq:def_cofactor_minor}\] We then define the cofactor matrix $\mC\in\R^{n\times n}$ of $\mA$  as the matrix of cofactors of all entries of $\mA$, i.e., $\mC = ((-1)^{i+j} \mM_{i,j})_{1\leq i, j\leq n}.$ The adjugate of $\mA$ is the transpose of the cofactor matrix $\mC$ of $\mA$, and hence its $(i,j)^\mathrm{th}$ entry may be expressed as: \[ \adj(\mA)_{ij} = (-1)^{i+j}\mM_{j,i}.\numberthis\label{eq:adj_in_terms_of_det}\] In particular, if the matrix $\mA$ is \emph{symmetric}, then the $(i,j)^\mathrm{th}$ minor equals the $(j,i)^\mathrm{th}$ minor, implying \[ \adj(\mA)_{ij} = \adj(\mA)_{ji} \text{ for all } i, j\in [n].\numberthis\label{eq:symm_mat_symm_adj}\] The minors of a matrix are also useful in computing its determinant. Specifically, the Laplace expansion of a matrix $\mA$ along its column $j$ is given as:\[ \det(\mA) = \sum_{i = 1}^n (-1)^{i+j} a_{ij} \mM_{i,j}.\numberthis\label{eq:laplace_exp_determinants}\] Finally, the adjugate $\adj(\mA)$ also satisfies the following important property:  \[\adj(\mA) \cdot \mA = \mA\cdot\adj(\mA) = \det(\mA) \cdot \mI.\numberthis\label{eq:adj_det_inv_connection}\] 
\end{fact}

\begin{fact}[Matrix determinant lemma, \cite{horn2012matrix}]\label{fact:block-matrix-determinant}
For any $\mM$, the determinant for a unit-rank update may be expressed as: \[ \det(\mM + \vu\vv^\top) = \det(\mM) + \vv^\top\adj(\mM) \vu.\] 
\end{fact}

\begin{fact}[Sherman-Morrison-Woodbury identity]\label{fact:shermanMorrisonWoodbury}
    Given conformable matrices $\mA, \mC, \mU, $ and $\mV$ such that $\mA$ and $\mC$ are invertible, we have 
\[ (\mA + \mU \mC \mV)^{-1} = \mA^{-1} - \mA^{-1} \mU (\mC^{-1} + \mV \mA^{-1} \mU )^-1 \mV \mA^{-1}.\]
\end{fact}

\noindent We crucially use  the following expansion for determinants of perturbed matrices. 

\begin{fact}[{\citet[Theorem $2.3$]{ipsen2008perturbation}}]\label{lem:det_A_Lambda}
Given $\mA \in \R^{\nconstr \times \nconstr}$
as in \Cref{fact:block-matrix-determinant}, positive diagonal matrix $\mLambda=\Diag(\boldsymbol{\lambda}) \in \R^{\nconstr \times \nconstr}$, and $\mA_{\bsigma}$  denoting the principal submatrix formed by selecting $\mA$'s rows and columns indexed by $\sigma\in \{0, 1\}^{\nconstr}$, we have
\[\det(\mA + \mLambda) = \sum_{\bsigma \in \{0,1\}^n}\left(\prod_{i=1}^m \lambda_i^{1 - \sigma_i}\right)\det(\mA_{\bsigma})\] 
\end{fact} 

\noindent We now state and prove a technical result that we build upon to prove \Cref{lem:split_into_adj}, which we in turn use in the proof of  \Cref{thm:convex_combination}.

\begin{lemma}\label{lem:adj_A}
Consider a matrix $\mA = \begin{bmatrix} 
a & \vb^\top \\ \vb &  \mD\end{bmatrix}\in \R^{n\times n},$ where $\mD\in \R^{(n-1)\times(n-1)}$ is a symmetric matrix. Then the adjugate $\adj(\mA)$ may be expressed as follows:
\[\adj (\mA) = \begin{bmatrix}
    \det(\mD) & -  \vb^\top \adj(\mD) \\
    -\adj(\mD) \vb & a\cdot \adj(\mD) + \mK
\end{bmatrix},\]
for some matrix $\mK$ independent of $a$. 
\end{lemma}
\begin{proof} Let $\mDmod_{ij}$ be  the $(n-2)\times (n-2)$ matrix obtained by deleting the $i^\mathrm{th}$ row and $j^\mathrm{th}$ column of $\mD$. Let $\mDmod_{j}$ be the $(n-1)\times (n-2)$ matrix formed by removing the $j^\mathrm{th}$ column of $\mD$, and let $\vbmod_i\in \R^{(n-2)}$ be the vector obtained by deleting the  $i^\mathrm{th}$ coordinate of $\vb$. With this notation in hand, we now compute some relevant cofactors. 

First, observe that $\mD$ is the matrix obtained by deleting the first row and column of $\mA$, and hence this fact along with  \Cref{eq:def_cofactor_minor} yields the   $(1,1)^\mathrm{th}$ cofactor:  
\[\mC_{1,1} = \det(\mD).\] 
Second, for some $j >0$, observe that the matrix obtained by deleting the first row of $\mA$ and the $(1+j)^\mathrm{th}$ column of $\mA$  is exactly the horizontal concatenation of $\vb$ and $\mDmod_j$. Applying this observation in \Cref{eq:def_cofactor_minor} then  gives the following expression for the $(1, 1+j)^\mathrm{th}$  cofactor:
\begin{align*}
    \mC_{1, \space 1+j} &= (-1)^{j}\det\left(\begin{bmatrix}
    \vb & \mDmod_j
    \end{bmatrix}\right)
    = (-1)^j \sum_{i=1}^{n-1} b_i (-1)^{i 
 + 1}\det(\mDmod_{ij}) 
    = -\sum_{i=1}^{n-1} b_i \adj(\mD)_{ij} 
    = -[\vb^\top \adj(\mD)]_j, 
\end{align*} where the second step is by using \Cref{eq:laplace_exp_determinants} to expand $\det\left(\begin{bmatrix} \vb & \mDmod_j\end{bmatrix}\right)$ along the column vector $\vb$, and the third step is by \Cref{eq:adj_in_terms_of_det} and \Cref{eq:symm_mat_symm_adj}, which applies since $\mD$ is assumed symmetric.
    Finally, to compute the $(1+i,1+j)^\mathrm{th}$ cofactor, we first construct the matrix obtained by deleting the $(1+i)^\mathrm{th}$ row and $(1+j)^\mathrm{th}$ column of $\mA$. Based on the notation we introduced above, this  may be expressed as $\begin{bmatrix}
        a & \vbmod_j^\top \\
        \vbmod_i & \mDmod_{ij}
    \end{bmatrix}$, from which we have by \Cref{eq:def_cofactor_minor}:
\begin{align*}
    \mC_{1 + i, \space 1 + j} &= (-1)^{i + j}\det\left(\begin{bmatrix}
        a & \vbmod_j^\top \\
        \vbmod_i & \mDmod_{ij}
    \end{bmatrix}\right),\numberthis\label{eq:cofactor_oneplusi_oneplusj}
\end{align*} 
 which we now simplify. To this end, we observe that \[ \det \left( \begin{bmatrix}
     a & \vbmod_j^\top \\ 
    \vbmod_j & \mDmod_{ij}
 \end{bmatrix} \right) = \det\left( \begin{bmatrix}
     a & \vbmod_j^\top \\ 
     0 & \mDmod_{ij} 
 \end{bmatrix}  + \begin{bmatrix}
     0 \\ \vbmod_j
 \end{bmatrix} \cdot \ve_1^\top \right) = \det\left( \begin{bmatrix}
     a & \vbmod_j^\top \\ 
     0 & \mDmod_{ij}
 \end{bmatrix}\right) + \ve_1^\top \adj\left( \begin{bmatrix}
     a & \vbmod_j^\top \\
     0 & \mDmod_{ij}
 \end{bmatrix} \right) \begin{bmatrix}
     0 \\ \vbmod_j
 \end{bmatrix},\numberthis\label{eq:cofactor_oneplusi_oneplusj_intermediate}\]
where we used \Cref{fact:block-matrix-determinant} in the second step. The first term in the right-hand side of the preceding equation may be simplified to $a \cdot \det(\mDmod_{ij})$. To simplify the second term, we first observe that we wish to compute only the first row of $\adj\left( \begin{bmatrix}
     a & \vbmod_j^\top \\
     0 & \mDmod_{ij}
 \end{bmatrix} \right)$; to this end, we introduce the notation that $\mX :=\mDmod_{ij}$ and $\vy = \vbmod_j$; we  denote $\mXt_{\ell j}$ to be the matrix obtained by deleting the $\ell^\mathrm{th}$ row and $j^\mathrm{th}$ column of $\mX$; we use $\mXt_{\ell}$ for the matrix obtained by deleting the $\ell^\mathrm{th}$ row of $\mX$. Now observe that for $\ell>0$, the $(1+\ell)^\mathrm{th}$ entry of the desired first row may be computed as follows:
\[ \adj\left( \begin{bmatrix}
    a & \vbmod_j^\top\\
    0 & \mDmod_{ij}
\end{bmatrix} \right)_{1, 1+\ell} = (-1)^{\ell} \det\begin{bmatrix}
    \vy^\top \\ 
   \mXt_{\ell}
\end{bmatrix} = \sum_{ j = 1}^{n-1} y_j \cdot(-1)^{\ell+j} \det\mXt_{\ell j}  = \sum_{j = 1}^{n-1} y_j \cdot (\adj\mXt)_{j\ell} = \vbmod_j^\top \adj(\mDmod_{i j})\ve_{\ell},\]
where the first step is by expressing the $(1, 1+\ell)^\mathrm{th}$ entry of the adjugate in question in terms of its $(1+\ell, 1)^\mathrm{th}$ minor (as per \Cref{eq:adj_in_terms_of_det}), the second step is by the Laplace expansion of the determinant along its first row (analogous to \Cref{eq:laplace_exp_determinants}), the third step is by \Cref{eq:adj_in_terms_of_det}, and the final step plugs back the newly introduced notation. Hence, we have \[\ve_1^\top \adj\left( \begin{bmatrix}
     a & \vbmod_j^\top \\
     0 & \mDmod_{ij}
 \end{bmatrix} \right) = \begin{bmatrix}a \cdot\det(\mDmod_{ij}) & \vbmod_j^\top \adj(\mDmod_{ij})\end{bmatrix}.\numberthis\label{eq:cofactor_computation_penultimate}\] Multiplying the right-hand side of \Cref{eq:cofactor_computation_penultimate} by $\begin{bmatrix}
     0 \\ \vbmod_j
 \end{bmatrix}$ and plugging the result back into \Cref{eq:cofactor_oneplusi_oneplusj_intermediate} and eventually into \Cref{eq:cofactor_oneplusi_oneplusj} then gives \[ \mC_{1+i, 1+j} = a\cdot \adj(\mD)_{ij} + (-1)^{i+j} \vbmod_j^\top \adj(\mDmod_{ij})\vbmod_j.\]
 By mapping these cofactors back into the definition of the adjugate we want, one can  then conclude the proof, with $\mK$ collecting all the $(-1)^{i + j}\vbmod_i^\top \adj(\mM_{ij})\vbmod_j$ terms.  
\end{proof}

\begin{corollary}\label{cor:adj_unit_rank_update}Let $\mA = \begin{bmatrix}a & \vb^\top \\ \vb & \mD\end{bmatrix}$ be a symmetric matrix. Then, 
\begin{align*}
    \adj(\mA + \lambda\ve_1 \ve_1^\top) &= \adj(\mA) + \lambda \begin{bmatrix} 0 & \zero^\top \\ \zero & \adj(\mD)\end{bmatrix}
\end{align*}
\end{corollary}
\begin{proof}
First, observe that by applying \Cref{lem:adj_A}, we have \[\adj\left(\begin{bmatrix}a + \lambda & \vb^\top \\ \vb & \mD\end{bmatrix}\right) = \begin{bmatrix}
        \det(\mD) & -\vb^\top \adj(\mD) \\
        -\adj(\mD)\vb & (a + \lambda) \cdot \adj(\mD) + \mK
\end{bmatrix}.\numberthis\label{eq:cor_intermediate}\] Next, observe that based on the definition of $\mA$, the left-hand side of \Cref{eq:cor_intermediate} is precisely $\adj(\mA+ \lambda \ve_1 \ve_1^\top)$; based on the expression for $\adj(\mA)$ from  \Cref{lem:adj_A}, the right-hand side of \Cref{eq:cor_intermediate} may be split into $ \adj(\mA) + \lambda \begin{bmatrix} 0 & 0 \\ 0 & \adj(\mD)\end{bmatrix}$, as desired. This concludes the proof. 
\end{proof}

\begin{lemma}\label{lem:svd_adj} 
Consider a matrix $\mL \in \R^{m\times n}$, and define the matrix $\mA = \mL\mL^\top \in  \R^{m \times m}$. Suppose that $\det(\mA) = 0$. 
Then, the following equation holds: 
\begin{align*}
    \adj(\mA)\mL  = \zero.
\end{align*}
\end{lemma}
\noindent To prove \Cref{lem:svd_adj}, we use the following two technical results from matrix analysis. 

\begin{fact}[Theorem $4.18$, \cite{laub2004matrix}]\label{fact:aapseudoL_L}
Suppose $\mA\in \R^{n\times p}$ and $\mL\in \R^{n\times m}$. Then $\mA\mA^\dagger\mL =\mL$ if and only if the range spaces $\mathcal{R}(\mL)$ and $\mathcal{R}(\mA)$ satisfy the inclusion 
 $\mathcal{R}(\mL)\subseteq \mathcal{R}(\mA)$.
\end{fact}

\begin{fact}[Theorem $3.21$, \cite{laub2004matrix}]\label{fact:range_L_range_LLt}
    Let $\mL\in \R^{m \times n}$. Then  the range spaces $\mathcal{R}(\mL)$ and $\mathcal{R}(\mL\mL^\top)$ satisfy the property $\mathcal{R}(\mL) = \mathcal{R}(\mL\mL^\top)$. 
\end{fact}

\begin{proof}[Proof of \Cref{lem:svd_adj}]
We prove the claim by showing that
\begin{align*}
    \adj(\mA)\mL &= [\adj(\mA)\mA]\mA^\dagger \mL = \zero,
\end{align*}
where the last equality follows from the property that $\adj(\mA)\mA = \det(\mA)\mI = 0$. All that remains is to prove that $\mL = \mA\mA^\dagger \mL$. By \Cref{fact:aapseudoL_L}, this is true if and only if $\mathcal{R}(\mL) \subseteq \mathcal{R}(\mA)$. From \Cref{fact:range_L_range_LLt}, we know that this is true. This concludes the proof. 
\end{proof}

\begin{lemma}\label{lem:Gtop_adjGHinvGtopsigma_zero}
Given a binary vector $\bsigma\in \left\{0, 1\right\}^m$, matrix  $\mG \in \R^{m \times n}$ and matrix $\mH \in \R^{n \times n}$ with the properties $\mH\succ0$ and $\det(\mG\mH^{-1}\mG^\top)_{\bsigma}=0$, we have \[ \mG^\top \adj(\mG\mH^{-1}\mG^\top)_{\bsigma}=0.\] 
\end{lemma}
\begin{proof}Without loss of generality, let $\mG = \begin{bmatrix} \mG_1 \\ \mG_2\end{bmatrix}$ where $\sigma_i = 1$ for the rows and columns associated with $\mG_2$. Then we may express $\mG\mH^{-1}\mG^\top$ in terms of $\mG_1$ and $\mG_2$ as follows: 
\begin{align*}
    \mG \mH^{-1} \mG^\top = \begin{bmatrix}
        \mG_1 \mH^{-1} \mG_1^\top & \mG_1 \mH^{-1} \mG_2^\top \\
        \mG_2 \mH^{-1} \mG_1^\top & \mG_2 \mH^{-1} \mG_2^\top
    \end{bmatrix}.
\end{align*} Based on the expansion above, observe that $(\mG\mH^{-1}\mG^\top)_{\bsigma} = \mG_2\mH^{-1}\mG_2^\top$. As a result, we may express $\mG^\top \adj(\mG\mH^{-1}\mG^\top)_{\bsigma}$, our matrix product of interest, as follows:
\begin{align*}
    \mG^\top \adj(\mG \mH^{-1}\mG^\top)_{\bsigma} = \begin{bmatrix}\mG_1^\top & \mG_2^\top\end{bmatrix}\begin{bmatrix} \zero & \zero \\ \zero & \adj(\mG_2 \mH^{-1} \mG_2^\top) \end{bmatrix} = \begin{bmatrix} \zero & \mG_2^\top \adj(\mG_2 \mH^{-1}\mG_2^\top)\end{bmatrix}.
\end{align*}
All that remains is to show that $\mG_2^\top \adj(\mG_2 \mH^{-1}\mG_2^\top) = \zero.$ To this end, we note that 
\begin{align*}
    \mG_2^\top \adj(\mG_2 \mH^{-1} \mG_2^\top) = \mH^{1/2} \left[\mH^{-1/2}\mG_2^\top \adj(\mG_2 \mH^{-1/2}\mH^{-1/2}\mG_2^\top)\right] = \zero,
\end{align*}
where the last equality follows immeditely by applying \Cref{lem:svd_adj}.
\end{proof}

\begin{lemma}\label{lem:split_into_adj} For a positive semi-definite matrix $\mA \in \R^{\nconstr \times \nconstr}$, a diagonal matrix $\mLambda = \Diag(\lambda)$, and the indicator vector $\bsigma\in \left\{0, 1\right\}^m$, define the following parameters:  
\begin{enumerate}[label=(\roman*)]
\compresslist{
\item Scaling factor $h_{\bsigma} = \det(\mA_{\bsigma})\prod_{i=1}^m \lambda_i^{1-{\bsigma}_i}$, 
\item Normalizing factor $h = \sum_{\bsigma \in \{0,1\}^m} h_{\bsigma}$, 
\item Adjugate scaling factor $c_{\bsigma} = \prod_{i=1}^m \lambda_i^{1-\sigma_i}$.
\item We split the set $\bsigma\in \left\{0, 1\right\}^m$ into the following two sets: \[ S := \left\{\bsigma \in \left\{0, 1\right\}^{\nconstr} \, \mid \, \det([\mG\mH^{-1}\mG^\top]_{\bsigma}) > 0\right\} \text{ and } S^{\complement} = \left\{\bsigma\in \left\{0, 1\right\}^m \, \mid \, \det([\mG\mH^{-1}\mG^\top]_{\bsigma})=0 \right\}.\] 

}
\end{enumerate} Assume that $\mA + \mLambda$ is invertible. Then we have the following decomposition of $(\mA + \mLambda)^{-1}$ in terms of inverses and adjugates of $\mA_{\bsigma}$ (the principal submatrices  of $\mA$), with the adjugate or inverse computed on the basis of whether  or not $\det(\mA_{\bsigma})=0$, as follows: 
\begin{align}
    (\mA + \mLambda)^{-1} &=  \sum_{\bsigma\in S} \frac{h_{\bsigma}}{h}\mA_{\bsigma}^{-1}
    + \sum_{\bsigma\in S^{\complement}} \frac{c_{\bsigma}}{h}\adj(\mA)_{\bsigma}.\label{eq:inverse_lemma}
\end{align}
\end{lemma}
\begin{proof}We begin by proving the following simpler statement for $c_{\bsigma} = \prod_{i = 1}^m \lambda_i^{1-\sigma_i}$: 
\begin{align}\label{eq:adj_split}
    \adj(\mA + \mLambda) &= \sum_{{\bsigma} \in \{0,1\}^m} c_{\bsigma} \adj(\mA)_{\bsigma}.
\end{align} Once this statement is proven, \Cref{eq:inverse_lemma} is implied by the following argument: Per \Cref{lem:det_A_Lambda}, we have that $\det(\mA + \mLambda) = \sum_{\bsigma\in \left\{0, 1 \right\}^m}h_{\bsigma} =  h$, so dividing throughout by $h$ yields $(\mA+\mLambda)^{-1}$ on the left-hand side (by \Cref{eq:adj_det_inv_connection}); the term $\sum_{\bsigma\in \left\{0, 1\right\}^m} \frac{c_{\bsigma}}{h}\adj(\mA)_{\bsigma}$ may be split into two sums of terms, one over those vectors $\bsigma\in \left\{0, 1\right\}^m$ for which  $\det(\mA_{\bsigma})=0$ and the second over those choices of $\bsigma\in \left\{0, 1\right\}^m$ for which $\det(\mA_{\bsigma})\neq0$. 
For terms such that $\det(\mA_{\bsigma}) \neq 0$, we have \[c_{\bsigma} \adj(\mA)_{\bsigma} = c_{\bsigma}\det([\mA]_{\bsigma})\frac{1}{\det([\mA]_{\bsigma})}\adj(\mA)_{\bsigma} = h_{\bsigma} \cdot
(\mA)_{\bsigma}^{-1}.\] Hence,   \Cref{eq:adj_split}, when divided by $h$, gives \Cref{eq:inverse_lemma}, as desired. 
We now prove \Cref{eq:adj_split}, proceeding via induction on $\nnz(\mLambda)$, the number of nonzero entries in $\mLambda$.
\\
\\
\textbf{Base case:}  When the number of non-zero entries $\nnz(\mLambda) = 0$,  by definition, $\mLambda = \Diag(\zero)$, which implies that the left-hand side of \Cref{eq:adj_split} is $\adj(\mA)$. Further, since by definition, $c_{\bsigma} = \prod_{i=1}^m \lambda_i^{1-\sigma_i}$, for our choice of  $\mLambda=\Diag(\zero)$, this gives the following expression:    \[c_{\bsigma}= \begin{dcases*}
0 & if  $\bsigma\neq \vec{1}$\,, \\[1ex]
1 & if $\bsigma=\vec{1}$\,.
\end{dcases*}
.\] With this choice of $c_{\bsigma}$, the right-hand side of \Cref{eq:adj_split} reduces to $\adj(\mA)$, which matches the left-hand side of \Cref{eq:adj_split}, thus implying that in this base case,  \Cref{eq:adj_split} is true. 
\\
\\
\textbf{Induction Step:} Suppose \Cref{eq:adj_split} holds for $\nnz(\mLambda) = k$. 
We now show  that \Cref{eq:adj_split} holds for $\nnz(\mLambda) = k + 1$ as well with some scaling factor $c_{\bsigma}$. 
 Without loss of generality, assume that $\lambda_i \neq 0$ for $i \in [k+1]$. Let $\mAmod_{11}$ be the $\R^{(m - 1) \times (m - 1)}$ matrix obtained by deleting the first row and first column of $\mA$. 
By expressing $\mA + \mLambda$ as $(\mA+ \sum_{i = 2}^{k+1} \lambda_i \ve_i\ve_i^\top) + \lambda_1 \ve_1\ve_1^\top$, we may use 
 \Cref{cor:adj_unit_rank_update} to expand $\adj(\mA + \mLambda)$ as follows: 
\begin{align*}
&\adj (\mA + \mLambda) 
= \adj\left(\mA + \sum_{i=2}^{k+1}\lambda_i \ve_i\ve_i^\top\right) + \lambda_1\begin{bmatrix}0 & \zero^\top \\ \zero & \adj\left(\mAmod_{11} + \sum_{i=1}^{k}\lambda_{i+1} \vemod_i\vemod_i^\top\right)\end{bmatrix},\numberthis\label{lem:induction_adj_intermediate}
\end{align*} where note that the $\ve_i\in \R^m$ and $\vemod_i\in\R^{m-1}$. 
We observe that both the terms on the right-hand side  have $\nnz(\mLambda) - 1 = k$ nonzero entries in their respective diagonal components. Hence, by our assumption, the induction hypothesis is applicable; therefore, suppose that by \Cref{eq:adj_split}, 
\begin{equation}\label{eq:adj_hyp_suppose} 
\begin{aligned} 
\adj(\mA + \sum_{i=2}^{\ell}\lambda_i \ve_i\ve_i^\top) &= \sum_{\sigma\in \left\{0, 1 \right\}^{m}}\hat{c}_\sigma \adj(\mA)_\sigma, \\
\adj(\mAmod_{11} + \sum_{i=1}^{\ell-1}\lambda_{i+1} \vemod_i\vemod_i^\top) &= \sum_{\sigma' \in \left\{0, 1 \right\}^{m-1}}\tilde{c}_{\sigma'} \adj(\mAmod_{11})_{\sigma'},
\end{aligned} 
\end{equation}
where, based on the diagonal components in each of the terms on the left-hand side, the scaling factors on the respective right-hand sides are $\hat{c}_\sigma = 0^{(1 - \sigma_1)}\prod_{i=2}^m \lambda_{i}^{1 - \sigma_{i}}$, 
 and $\tilde{c}_{\sigma'} = \prod_{i=1}^{m-1} \lambda_{i+1}^{1 - \sigma'_{i}}$. As a consequence of these definitions, we can re-write the terms in \Cref{eq:adj_hyp_suppose} using $c_{\bsigma}$ as follows. First, observe that $\hat{c}_\sigma = 0$ when $\sigma_1 = 0$ and $\hat{c}_\sigma = c_\sigma$ otherwise. This implies: 
 \[ \sum_{\bsigma\in \left\{0, 1\right\}^m} \tilde{c}_{\bsigma} \adj(\mA)_{\bsigma} = \sum_{\substack{\bsigma\in \left\{0, 1\right\}^m,\\ \sigma_1 = 1}} {c}_{\bsigma} \adj(\mA)_{\bsigma}. \numberthis\label{lem:induction_adj_1} \] Next, observe that for the vector $\bsigma= [0; \bsigma']$ formed by concatenating zero with $\bsigma'$, we have $c_{\bsigma} = \lambda_1 \tilde{c}_{\bsigma'}$. This implies the following chain of equations: 
 \[ \lambda_1\begin{bmatrix} 0 & \zero^\top \\ \zero & \sum_{\bsigma'\in \left\{0, 1 \right\}^{m-1}} \tilde{c}_{\bsigma'} \adj(\mAmod_{11})_{\bsigma'}\end{bmatrix} = 
 \sum_{{\bsigma}' \in \{0,1\}^{m-1}}\lambda_1 \tilde{c}_{{\bsigma}'} \adj(\mA)_{\begin{bsmallmatrix}0 \\ {\bsigma}'\end{bsmallmatrix}} =  \sum_{\substack{{\bsigma} \in \{0,1\}^m, \\ {\bsigma}_{0} = 0}}c_{\bsigma} \adj(\mA)_{\bsigma},\numberthis\label{lem:induction_adj_2}\] where in the first step, we used the fact that $\mAmod_{11}$ is, by definition, the principal submatrix of $\mA$ obtained by deleting its first row and first column; the second step is by our prior observation connecting $c_{\bsigma}$ and $\tilde{c}_{\bsigma'}$. 
Plugging the right-hand sides from \Cref{eq:adj_hyp_suppose} into that of \Cref{lem:induction_adj_intermediate} and then applying \Cref{lem:induction_adj_1}
 and \Cref{lem:induction_adj_2} gives
\begin{align*}
    \adj \left(\mA + \mLambda \right) 
    &= \sum_{{\bsigma} \in \{0,1\}^m}\hat{c}_{\bsigma} \adj(\mA)_{\bsigma} + \lambda_1\begin{bmatrix}0 & \zero^\top \\ \zero & \sum_{{\bsigma}' \in \{0,1\}^{m-1}}  \tilde{c}_{{\bsigma}'}\adj(\mAmod_{11})_{{\bsigma}'}\end{bmatrix} \\
    &= \sum_{\substack{{\bsigma} \in \{0,1\}^m, \\ {\bsigma}_{0} = 1}}c_{\bsigma} \adj(\mA)_{\bsigma} + \sum_{\substack{{\bsigma} \in \{0,1\}^m, \\ {\bsigma}_{0} = 0}}c_{\bsigma} \adj(\mA)_{\bsigma} \\ 
    &= \sum_{{\bsigma} \in \{0,1\}^m} c_{\bsigma} \adj(\mA)_{\bsigma}.
\end{align*}  
Thus, we have shown   \Cref{eq:adj_split} for $\nnz(\mLambda)=k+1$, thereby completing the induction and concluding the proof of \Cref{eq:adj_split} and, consequently, of the stated lemma. 
\end{proof}

\section{Technical Results from Convex Analysis}\label{sec:convex_analysis_lemmas}
\begin{fact}[\cite{nesterov1994interior}]\label{thm:inner_prod_ub_nu}
    Let $\Phi$ be a $\nu$-self-concordant barrier. Then for any $\vx\in \textrm{dom}(\Phi)$ and $\vy\in \textrm{cl(dom)}(\Phi)$,  \[\nabla \Phi(\vx)^\top (\vy-\vx) \leq \nu.\]  
\end{fact}

\begin{fact}[{\cite[Proposition $2.3.1$]{nesterov1994interior}}]\label{thm:affine_scb} 
Let $G$ be a closed convex domain in $E$, let $F$ be a $\nu$-self-concordant barrier for $G$, and let $x=\mathcal{A}(y)$ be an affine transformation from a space $E^{\prime}$ into $E$ such that $\mathcal{A}(E^\prime)\cap \mathrm{int}(G)\neq \emptyset$. Let $G^\prime=\mathcal{A}^{-1}(G)$ and $F^\prime(y)= F(\mathcal{A}(y)):\mathrm{int}(G^\prime)\to \mathbb{R}$. Then, $F^\prime$ is a $\nu$-self-concordant barrier for $G^\prime$. 
\end{fact}

To prove \cref{{thm:res_lower_bound}}, we use two simple properties of self-concordant barriers \cref{{lem:f_sc_one_by_nine_r_squared}} and \cref{{lem:linear_plus_barrier_sc}} that originally appeared in \cite{ghadiri2024improving}. We provide their proofs here for completeness. 

\begin{lemma}[\cite{ghadiri2024improving}]
\label{lem:f_sc_one_by_nine_r_squared}If $f$ is a self-concordant barrier for a set $\calK\subseteq \calB(0,R),$
then for any $x\in \calK$, we have  \[\nabla^{2}f(\vx)\succeq\frac{1}{9R^{2}}I.\] 
\end{lemma}
\begin{proof}
For the sake of contradiction, suppose $\nabla^{2}f\not\succeq\frac{1}{9R^{2}}I.$ This is equivalent to, for some $x\in \calK$ and unit vector $\u$, 
\begin{equation}
(3R\u)^{\top}\nabla^{2}f(\vx)(3R\vu)<1.\label[ineq]{eq:contradiction_ineq}
\end{equation}
Define the unit-radius
Dikin ellipsoid around $\vx$ as \[\mathcal{E}_\vx(\vx, 1)=\left\{\vy:(\vy-\vx)^{\top}\nabla^{2}f(\vx)(\vy-\vx)\leq1\right\}.\]
Then, \Cref{eq:contradiction_ineq} is equivalent to the assertion that $\vx+3R\vu\in\mathcal{E}_\vx(\vx,1).$
Because $f$ is self-concordant we have $\mathcal{E}_\vx(\vx,1)\subseteq \calK$ (see, e.g., \cite[Theorem $2.1.1$]{nesterov1994interior}).
This, combined with $\vx+3R\vu\in\mathcal{E}_\vx(\vx,1)$,  implies $\vx+3R\vu\in \calK.$
However, since $\calK\subseteq \calB(0,R)$ and $x\in \calK$ by construction,
the inclusion $\vx+3R\vu\in \calK$  cannot hold for any unit vector $\vu,$ which implies that our initial assumption must be false, thus concluding the proof.  
\end{proof}
\begin{lemma}[\cite{ghadiri2024improving}]
\label{lem:linear_plus_barrier_sc}If $f$ is a $\nu$-self-concordant barrier for a given
convex set $\calK$ then $g(\vx)=c^{\top}\vx+f(\vx)$ is a self-concordant barrier over $\calK$. Further, if $\calK\subseteq \calB(0,R),$ then $g$ has self-concordance parameter at most 
\[20(\nu+R^{2}\|\vc\|^{2}).\] 
\end{lemma}
\begin{proof}
Since $\nabla^{2}g=\nabla^{2}f,$ we can conclude that $g$ is also
a self-concordant function. Since $\calK\subseteq \calB(0,R)$, \Cref{lem:f_sc_one_by_nine_r_squared}
applies, and we have $\nabla^{2}f(\vx)\succeq\frac{1}{9R^{2}}\mI\text{ for all }x\in \calK.$
Equivalently, 
\[
\nabla^{2}f(\vx)^{-1}\preceq9R^{2}\mI\text{ for all }x\in \calK. \numberthis\label[ineq]{ineq:nabla_squared_f_at_most_squared_R}
\]
 The self-concordance parameter  (see \Cref{def:sc_and_scb}) of $g$ is:  
\begin{align*}
\|\nabla g(\vx)\|_{\nabla^{2}g(\vx)^{-1}}^{2}  =\|\vc+\nabla f(\vx)\|_{\nabla^{2}f(\vx)^{-1}}^{2} \leq2\|\vc\|_{\nabla^{2}f(\vx)^{-1}}^{2}+2\|\nabla f(\vx)\|_{\nabla^{2}f(\vx)^{-1}}^{2},\numberthis\label[ineq]{ineq:sc_linear_plus_barr_last}
\end{align*}
 where 
 the first step is by definition of self-concordance parameter of $g$. 
 To finish the proof, we recall that 
$\|\vc\|_{\nabla^{2}f(\vx)^{-1}}^{2}\leq 9R^{2}\|\vc\|_{2}^{2}$ 
by \Cref{ineq:nabla_squared_f_at_most_squared_R}, and 
$\|\nabla f(\vx)\|_{\nabla^{2}f(\vx)^{-1}}^{2}\leq
\nu$ 
 by the  self-concordance parameter of $f$ and put these bounds into \Cref{ineq:sc_linear_plus_barr_last}. 
\end{proof}

Both our proofs of our main result lower bounding the residual crucially build upon the following result from \cite{zong2023short}, who proved the following highly non-trivial lower bound in \cref{{eq:2}} (the upper bound was known in the classical literature on interior-point methods). 
\begin{fact}[\cite{zong2023short}, Theorem 2]\label{lem:zong_opt_gap}
    Fix a vector $\vc$, a polytope $\calK$, and a point $\v.$ We assume
that the polytope $\calK$ contains a full-dimensional ball of radius $r.$ Let $\vstar=\arg\min_{\u\in \calK}\vc^{\top}\u$.
We define, for $\vc,$  
\begin{equation}
\gap(\v)=\vc^{\top}(\v-\vstar).\label{eq:1}
\end{equation}
Further, define $\v_{\eta}=\arg\min_{\v}\vc^{\top}\v+\eta\phiK(\v)$,
where $\phiK$ is a self-concordant barrier on $\calK.$ Then we have the following
lower bound on this suboptimality gap evaluated at $\veta$: 
\begin{equation}
\min\left\{ \frac{\eta}{2},\frac{r\|\vc\|}{2\nu+4\sqrt{\nu}}\right\} \leq\gap(\veta)=\vc^{\top}(\veta-\vstar) \leq \eta\nu.\label[ineq]{eq:2}
\end{equation}
\end{fact}

Finally, we need the following known claim about the minimizer of a convex combination of two convex functions. We provide its proof for the sake of completeness of this document.  
\begin{claim}\label{claim:one_dim_grad_monotone}
Let $f:\mathbb{R\to\mathbb{R}}$ and $g:\mathbb{R}\to\mathbb{R}$
be two convex, continuously differentiable functions with
$x_{f}$ and $x_{g}$ satisfying $f^{\prime}(x_{f})=0$ and $g^{\prime}(x_{g})=0.$ Additionally, let $g$ be strictly convex. 
Let $x_{f}<x_{g}$. Then, for any positive $\alpha$ and $\beta,$
we have that the point $x_{\alpha f+\beta g}$ satisfying $\alpha f^{\prime}(x_{\alpha f+\beta g})+\beta g^{\prime}(x_{\alpha f+\beta g})=0$ additionally
satisfies $x_{\alpha f+\beta g}\in[x_{f},x_{g}].$ 
\end{claim}
\begin{proof}
Since $f$ and $g$ are both  convex, their gradients are
monotone (and the gradients of $g$ are strictly monotone). Since $x_f < x_g$, this then implies $0=f^{\prime}(x_{f})\leq f^{\prime}(x_{g})$
and $g^{\prime}(x_{f})<g^{\prime}(x_{g})=0.$ Multiplying the first
inequality by $\alpha$ and the second by $\beta$ and summing them
gives $\alpha f^{\prime}(x_{f})+\beta g^{\prime}(x_{f})<0<\alpha f^{\prime}(x_{g})+\beta g^{\prime}(x_{g}).$
By the mean value theorem, there must be a point $x_{\alpha f+\beta g}\in[x_{f},x_{g}]$
at which the derivative is zero. Since $\alpha f + \beta g$ is strictly convex, the point where its derivative is zero must be unique. This concludes the proof. 
\end{proof}

\subsection{A First Lower Bound on the Residual}
Equipped with the tools from the previous sections, we now provide a preliminary lower bound on the residual, as desired to  claim  smoothness of our barrier MPC solution.  
\begin{theorem}\label{thm:res_lower_bound}
Let $\mathcal{K}=\left\{ \vx:{{A}}\vx\geq{b}\right\} $ be
a polytope such that each of $m$ rows of ${A}$ is normalized to be
unit norm. Let $\mathcal{K}$ contain a ball of radius $r$ and be contained
inside a ball of radius $R$ centered at the origin. Let 
\begin{equation}
\ue:=\arg\min_{\u}q(\u)+\eta\phi_{\mathcal{K}}(\u),\label{eq:def-ueta}
\end{equation}
where $q$ is a convex $L$-Lipschitz function and $\phi_{\mathcal{K}}$ is a $\nu$-self-concordant barrier on $\mathcal{K}$. We show for $\re_i(u_\eta)$, the $i^\mathrm{th}$ residual at $u_\eta$, the following lower bound:
 \[\re_i(u_\eta)\geq  \min\left\{ \frac{\eta}{2}, \frac{r\eta^2}{150 (\nu\eta^2 + R^2 (L^2 +1) )} \right\}.\] 
\end{theorem}
\begin{proof}[Proof of \Cref{thm:res_lower_bound}]
Applying the first-order optimality condition of $\ue$ in \Cref{eq:def-ueta} gives us that 
\begin{equation}
\eta\nabla\phi_{\mathcal{K}}(\ue)+\nabla q(\ue)=0.\label{eq:a}
\end{equation}
From here on, we fix $\vc=\nabla q(\ue),$where $\ue$ is as in
\Cref{eq:def-ueta}. Then, 
 we may conclude 
\begin{equation}
\ue\in\arg\min_{\u}\vc^{\top}\u+\eta\phi_{\mathcal{K}}(\u),\label{eq:4}
\end{equation}
where we have
replaced the cost $q$ in   \Cref{eq:def-ueta} with a specific linear cost $\vc$; to see \Cref{eq:4},
observe that $\ue$ satisfies the first-order optimality condition
of  \Cref{eq:4} because of  \Cref{eq:a} and our choice of $\vc$. 

We now define the function $\phitilde_{\mathcal{K}}(x)= \eta^{-1}\cdot(c - a_i)^\top x + \phi_{\mathcal{K}}(x)$. By \Cref{lem:linear_plus_barrier_sc}, we have that $\phitilde_{\mathcal{K}}$ is a self-concordant-barrier on ${\mathcal{K}}$ with self-concordance parameter 
\[ \widetilde{\nu}\leq 20 (\nu+ R^2\eta^{-2}(\|c\|^2+\|a_i\|^2).\numberthis\label{eq:new-scb-parameter} \] With this new self-concordant barrier in hand, we may now express $\vueta$ from \Cref{eq:4} as the following optimizer: \[ \vueta = \arg\min_u a_i^\top u + \eta \phitilde_{\mathcal{K}}(u).\numberthis\label{eq:new-expression-for-ueta}\] Further,
let $\us\in\arg\min_{\u\in {\mathcal{K}}}a_i^\top\u$. 
By applying 
\Cref{lem:zong_opt_gap} to $\vueta$ expressed as in  \Cref{eq:new-expression-for-ueta}, we have 
\begin{equation}
\min\left\{ \frac{\eta}{2},\frac{r\|a_i\|}{2\widetilde{\nu}+4\sqrt{\widetilde{\nu}}}\right\} \leq a_i^{\top}(\ue-\us).\label[ineq]{eq:dist_bound_zeroth}
\end{equation} The lower bound in \Cref{eq:dist_bound_zeroth} may be expanded upon via \Cref{eq:new-scb-parameter}, and chaining this with the observation  $a_i^\top (u_\eta- \us) = \textrm{res}_i (\vueta) - \textrm{res}_i(\us)$ gives: \[ \min\left\{ \frac{\eta}{2}, \frac{r\|a_i\|}{150 (\nu + R^2 \eta^{-2} (\|c\|^2 + \|a_i\|^2) )} \right\} \leq \textrm{res}_i(\vueta)-\textrm{res}_i(\us).\] The definition of $\us$ 
implies $\textrm{res}_i(\us) \geq 0$, hence  $\textrm{res}_i(u_\eta)\geq  \min\left\{ \frac{\eta}{2}, \frac{r}{150 (\nu + R^2 \eta^{-2} (L^2 + 1) )} \right\}.$ Repeating this computation for each constraint of $\mathcal{K}$ gives the claimed bound overall. 
\end{proof}

\subsection{An Improved Lower Bound on the Residual}
In this section, we strengthen the bound from \cref{thm:res_lower_bound} via a more careful analysis. 
\subsubsection{Warmup: The One-Dimensional Case}
We begin with a lemma on optimizing quadratics in one dimension to motivate our later results for arbitrary polytopes in higher dimensions.

\begin{lemma}\label{lem:quad_gap} Let $\phi$ be a $\nu$-self-concordant barrier over $(0, r)$ and $q$ be a convex function such that $\nabla q(v) = 0$ and $0 < m \leq \nabla^2 q(x) \leq M$. Define,
\begin{align*}
    x^\eta := \arg\min_{x} q(x) + \eta \phi(x).
\end{align*}
Then,
\begin{align}
\min\left\{\frac{1}{2}\left(\sqrt{\frac{2\eta}{M} + v^2} + v\right), \frac{mr}{M(2\nu + 4 \sqrt{\nu})}\right\} \leq x^\eta \leq \frac{1}{2}\left(\sqrt{\frac{4\eta\nu}{m} + v^2} + v\right).\label{eq:x_eta_bound}
\end{align}
\end{lemma}

\begin{proof} Let $c := \nabla q(x^{\eta})$. Via the same trick as for \cref{{eq:4}},  we can express $\vx^\eta$ equivalently as $x^\eta = \arg\min_{x} c x + \eta \phi(x)$. Let $x^\star := \arg\min_{x \in (0,r)} cx$ and $\tilde{x} := \arg\min_{x} \phi(x)$.
\\ \\
\textbf{Case 1}: $v < \tilde{x}$. Then by \cref{claim:one_dim_grad_monotone} applied to the functions $q$ and $\phi$, we may deduce that $v < x^\eta < \tilde{x}$, meaning $c > 0$ and therefore $x^\star = 0$.
\\\\
Applying \Cref{lem:zong_opt_gap}, we have:
\begin{align*}
    \min \left\{ \frac{\eta}{2}, \frac{r c}{2 \nu + 4 \sqrt{\nu}} \right\} \leq cx^\eta \leq \eta\nu.
\end{align*}
Using that $m(x^\eta - v) \leq c \leq M(x^\eta - v)$ we have,
\begin{align*}
    \min \left\{ \frac{\eta}{2 M}, \frac{r(x^\eta - v)}{2\nu + 4 \sqrt{\nu}} \right\} &\leq (x^\eta - v)x^\eta \leq \frac{\eta\nu}{m}.
\end{align*}
Solving $\frac{\eta}{2M} \leq (x^\eta - v)x^\eta \leq \frac{\nu \eta}{m}$ with the condition that $x^\eta > v$, we have, 
\begin{align*}
\frac{1}{2}\left(\sqrt{\frac{2\eta}{M} + v^2} + v\right) \leq x^\eta \leq \frac{1}{2}\left(\sqrt{\frac{4\eta\nu}{m} + v^2} + v\right).
\end{align*}
Combining with the minimum on the LHS bound, we arrive at
\begin{align*}
\min\left\{\frac{1}{2}\left(\sqrt{\frac{2\eta}{M} + v^2} + v\right), \frac{r}{2\nu + 4 \sqrt{\nu}}\right\} \leq x^\eta \leq \frac{1}{2}\left(\sqrt{\frac{4\eta\nu}{m} + v^2} + v\right).
\end{align*}
\\
\textbf{Case 2}: $v \geq \tilde{x}$. By the same reasoning as in the previous case, we have $\tilde{x} \leq x^\eta \leq v$. Note that by applying \Cref{lem:zong_opt_gap} with $c = 1$ and considering $\eta \to \infty$, we can deduce that $\tilde{x} \geq \frac{r}{2\nu + 4\sqrt{\nu}}$. We can see that \Cref{eq:x_eta_bound} still holds as,
\begin{align*}
    \min\left\{\frac{1}{2}\left(\sqrt{\frac{\eta}{M} + v^2} + v\right), \frac{r}{(2\nu + 4 \sqrt{\nu})}\right\} \leq \frac{r}{2\nu + 4\sqrt{\nu}} \leq \tilde{x} \leq x^\eta \leq v \leq \frac{1}{2}\left(\sqrt{\frac{4\eta\nu}{m} + v^2} + v\right).
\end{align*}
\end{proof}

The above result shows that if the minimizer of a strongly convex cost lies outside of the constraint set, we should expect to get a 
bound of the form $O(\sqrt{\eta + v^2} - v)$, where $v$ is the distance to the constraint set. 
\subsubsection{Upper Bounds on Approximation Error}
\begin{lemma}\label{lem:iso_quad_universal_upper_bound}Let $\mathcal{K} \subset \R^n$ be a polytope and $\phi$ be a $\nu$-self-concordant barrier on $\mathcal{K}$. Let $\vx^\eta := \argmin_{\vx} \frac{\alpha}{2}\|\vx - \vv\|^2 + \eta \phi(\vx)$ and  $\vx^\star := \argmin_{\vx \in \mathcal{K}} \frac{\alpha}{2}\|\vx - \vv\|^2$ for some $\v \in \R^n$. Then,
\begin{align*}
   \left\|\vx^\eta - \vx^\star\right\| \leq \sqrt{\frac{\eta\nu}{\alpha}}.
\end{align*}
\end{lemma}
\begin{proof}
We proceed similar to \Cref{thm:error_bound_barrier_mpc}. Note that by $\alpha$-strong-convexity of $q(\vx) := \frac{\alpha}{2}\|\vx - \vv\|^2$, we have that,
\begin{align*}
    [\nabla q(\vx^\eta) - \nabla q(\vx^\star)]^\top (\vx^\eta - \vx^\star)  \geq \alpha\|\vx^\eta - \vx^\star\|^2.
\end{align*}
Note that from the optimality condition $\nabla q(\vx^\eta) + \eta \nabla \phi(\vx^\eta) = 0$, and, by convexity of $\mathcal{K}$, $\nabla q(\vx^\star)^\top[\vx^\eta - \vx^\star] \geq 0$, and by \Cref{thm:inner_prod_ub_nu}, it follows,
\begin{align*}
    \alpha \|\vx^\eta - \vx^\star\|^2 \leq \eta\nabla \phi(\vx^\eta)^\top [\vx^\star - \vx^\eta] - \nabla q(\vx^\star)^\top [\vx^\eta - \vx^\star] \leq \nabla\phi(\vx^\eta)^\top [\vx^\star - \vx^\eta] \leq \eta\nu.
\end{align*}
Rearranging the terms then gives the claim. 
\end{proof}
Note that the above result can be generalized to $\alpha$-strongly-convex functions. 
In the next lemma, we show that we can make a similar bound along the gradient of the cost function at $\vx^\star$.
\begin{lemma}\label{lem:iso_quad_upper_bound}Let $\mathcal{K} \subset \R^n$ and $\phi$ be a $\nu$-self-concordant barrier on $\mathcal{K}$. Let $\vx^\eta := \argmin_{\vx} \frac{\alpha}{2}\|\vx - \vv\|^2 + \eta \phi(\vx)$ and $\vx^\star := \argmin_{\vx \in \mathcal{K}} \frac{\alpha}{2}\|\vx - \vv\|^2$ for some $\vv \in \R^{n}$. Assume that $\vx^\star \neq \vv$, and let $\va = \frac{\vx^\star - \vv}{\|\vx^\star - \vv\|}$. Then,
\begin{align*}
    0 \leq \va^\top (\vx^\eta - \vx^\star) \leq \frac{1}{2}\left(\sqrt{\frac{4\eta\nu}{\alpha} + \|\vx^\star - \vv\|^2} - \|\vx^\star -\vv\|\right).
\end{align*}
(Note that for the case where $\vx^\star = \vv$, \Cref{lem:iso_quad_universal_upper_bound} can be chosen for any $\va$)
\end{lemma}

\begin{proof}Let $q(\vx) := \frac{\alpha}{2}\|\vx-\vv\|^2$. Note that we can write:
\begin{align*}
    \nabla q(\vx^\eta) = \alpha  \cdot \va^\top (\vx^\eta - \vx^\star) \va + \alpha \cdot \vb^\top (\vx^\eta - \vx^\star)\vb + \nabla q(\vx^\star),
\end{align*}
where $\|\va\| = \|\vb\| = 1$ and $\vb \perp \va$. Using \cref{{eq:a}} (valid because in this lemma we define $\vx^\eta:= \arg\min_x q(\vx) + \eta\phi(\vx)$ for a convex $q$) with  \cref{{thm:inner_prod_ub_nu}}  yields
\begin{align*}
    \nabla q(\vx^\eta)^\top(\vx^\eta - \vx^\star) \leq \eta\nu.
\end{align*}
Then it follows that
\begin{align*}
    \alpha \cdot [\va^\top(\vx^\eta - \vx^\star)]^2 + \alpha [\vb^\top (\vx^\eta - \vx^\star)]^2 + \alpha \|\vx^\star - \vv\| \cdot \va^\top (\vx^\eta - \vx^\star) \leq \eta\nu.
\end{align*}
We  drop  $\alpha \cdot [b^\top (\vx^\eta - \vx^\star)]^2$  and  solve for $\va^\top (\vx^\eta - \vx^\star)$ to prove our claimed upper bound on $ \va^\top (\vx^\eta - \vx^\star)$.  To derive the lower bound, 
 we use that $0 \leq \nabla q(\vx^\star)^\top [\vx^\eta - \vx^\star] = c \cdot \va^\top[\vx^\eta - \vx^\star]$ for some $c > 0$.
\end{proof}

\subsubsection{An Improved Lower Bound on the Residual}
\begin{lemma}\label{lem:contains_linear_ball} Fix a polytope $\mathcal{K}$, a convex function $q$, and a $\nu$-self-concordant barrier $\phi$ over $\mathcal{K}$. Assume that the polytope $\mathcal{K}$ contains a full-dimensional ball of radius $r$ and is contained within a ball of radius $R$ around some point $\bar{\vx}$, i.e. $\mathcal{B}(\bar{\vx},r) \subseteq \mathcal{K} \subseteq \mathcal{B}(\bar{\vx}, R)$. Let $\vx^\eta := \arg\min q(\vx) + \eta\phi(\vx)$ for arbitrary $\eta > 0$,
\begin{align*}
    \mathcal{B}\left(\vx^\eta, \frac{r}{R}\min\left\{\frac{\eta}{2 \|\nabla q(\vx^\eta)\|}, \frac{r}{2\nu + 4 \sqrt{\nu}}\right\}\right) \subseteq \mathcal{K}. \numberthis\label{eq:eta_ball_around_x_eta}
\end{align*}
\end{lemma}

\begin{proof} 
Consider the line passing through $\bar{\vx}$ and $\vx^\eta$ given by $\mathcal{S} = \{\bar{\vx}t + \vx^\eta(1-t) : t\}$ and let $\vx_1$ and  $\vx_2$ be the endpoints of $\mathcal{K} \cap \mathcal{S}$.  Equipped with these definitions, we will show
\begin{align*}\min(\|\vx^\eta - \vx_1\|,\|\vx^\eta - \vx_2\|) \geq \min \left\{\frac{\eta}{2\|\nabla q(\vx^\eta)\|}, \frac{r}{2\nu + 4 \sqrt{\nu}}\right\}.\numberthis\label[ineq]{ineq:dist_to_xeta_xo_xt}
\end{align*} Before proving \cref{{ineq:dist_to_xeta_xo_xt}}, we first show why it immediately gives the claimed result of \cref{eq:eta_ball_around_x_eta}. Pick $\hat{\vx} \in \{\vx_1, \vx_2\}$ such that $\vx^\eta$ lies on the line segment between $\hat{\vx}$ and $\bar{\vx}$. Consider any direction $\vc$ such that  $\|\vc\| = 1$. Consider the triangle formed by the points $\hat{\vx}, \bar{\vx}$, and $\bar{\vx} + r \vc$, and draw a line segment from $\vx^\eta$, parallel to $c$ and intersecting the line segment from $\hat{\vx}$ to $\bar{\vx}$ at a point we label $y$. Then, by convexity, $y\in \mathcal{K}$. Then, we  prove \cref{eq:eta_ball_around_x_eta} by showing that $\|y-\vx^\eta\|\geq \frac{r}{R}\min\left\{\frac{\eta}{2 \|\nabla q(\vx^\eta)\|}, \frac{r}{2\nu + 4 \sqrt{\nu}}\right\}$. To see this inequality, we note that \[ \|y-\vx^\eta\|= \|\bar{\vx} -(\bar{\vx} + r c)\| \cdot \frac{\|\vx^\eta-\hat{\vx}\|}{\|\bar{\vx} - \hat{\vx}\|} = r\cdot \frac{\|\vx^\eta-\hat{\vx}\|}{\|\bar{\vx} - \hat{\vx}\|}  \geq r \cdot \frac{1}{R} \cdot  \min \left\{\frac{\eta}{2\|\nabla q(\vx^\eta)\|}, \frac{r}{2\nu + 4 \sqrt{\nu}}\right\},\] where the first equation is by similarity of the  triangles formed by $\hat{\vx}$, $\bar{\vx}$, and $\bar{\vx}+rc$ and by $\hat{\vx}$, $\vx^\eta$, and $y$; the final step is by  \cref{ineq:dist_to_xeta_xo_xt} and the assumed upper bound of $R$ on the polytope diameter. 
We now proceed to prove \cref{ineq:dist_to_xeta_xo_xt}.  
Without loss of generality, let $\vx_1$ be such that $\nabla q(\vx^\eta)^\top \left(\frac{\bar{\vx} - \vx_1}{\|\bar{\vx} - \vx_1\|}\right) \geq 0$. We use this characterization of $\vx_1$ in both parts of our proof below. 

\paragraph{Lower bound on $\|\vx^\eta-\vx_2\|$.} Denote the restriction of the barrier $\phi$ (defined on the polytope $\mathcal{K}$) to the line  $\mathcal{S}$ by a univariate function $\xi$, so that \[\xi(t):= \phi(t\vx_2 + (1-t)\vx^\eta), \text{ with } \xi^\prime(t)= \nabla\phi(t\vx_2+(1-t)\vx^\eta)^\top (\vx_2 - \vx^\eta).\numberthis\label{eq:def_univariate_restriction_phi}\]  We note that $\xi(0)= \phi(\vx^\eta)$ and  $\xi(1)=\phi(\vx_2)$. By definition of $\phi$ as a barrier on $\mathcal{K}$, note that $\xi$ is also a barrier defined only on  $\mathcal{K}\cap \mathcal{S}$~\cite{renegar2001mathematical}. 
Define the following  quantities associated with $\xi$: \[t_{\mathrm{ac}}:=\arg\min_t \xi(t), \text{ and } \vx_{\mathrm{ac}} := t_{\mathrm{ac}}\vx_2 + (1-t_{\mathrm{ac}}) \vx^\eta.\numberthis\label{def:x_ac}\] In other words, $\vx_{\mathrm{ac}}$ is the analytic center of the barrier $\xi$ on $\mathcal{S}$. We now apply \Cref{lem:zong_opt_gap} with $c = \vx_{\mathrm{ac}} - \vx_2$, $\mathcal{K}=\mathcal{S}$, the barrier $\xi$ on $\mathcal{S}$,   $\eta \to \infty$, and denoting $\nu_{\xi}$ to be the self-concordance parameter of $\xi$. Then, combining   the lower bound in \cref{lem:zong_opt_gap} with the definition of $\vx_{\mathrm{ac}}$ and $\nu_{\xi}\leq \nu$~\cite{renegar2001mathematical} yields \[\|\vx_{\mathrm{ac}} - \vx_2\| \geq \frac{r}{2\nu_{\xi} + 4\sqrt{\nu_{\xi}}} \geq \frac{r}{2\nu + 4\sqrt{\nu}}.\numberthis\label[ineq]{ineq:xac_xtwo_lb}\]  Next we have  by  the choice of $\vx_1$ that $\nabla q(\vx^\eta)^\top [\bar{\vx} - \vx_1] \geq 0$. By the first-order optimality condition of $\vx^\eta$ (as in \cref{{eq:a}}), this is equivalent to $ \nabla \phi(\vx^\eta)^\top[\bar{\vx} - \vx_1]\leq 0$. This final inequality implies \[\nabla \phi(\vx^\eta)^\top[\vx_2 - \vx^\eta] \leq 0\numberthis\label[ineq]{ineq:grad_xeta_inner_xtwo_minus_xeta}\] since $\vx_1$ and $\vx_2$ are the end points of $\mathcal{K}\cap \mathcal{S}$ (the line segment whose interior contains $\bar{\vx}$ and $\vx^\eta$), and hence $\vx_2-\vx^\eta$ is a vector in the same direction as $\bar{\vx}-\vx_1$. 
Note that we have
\[\xi^\prime(0)\leq 0 \text{ and } \xi^\prime (t_{\mathrm{ac}})=0,\numberthis\label[ineq]{eq:psi_prime_zero_negative}\] where the first inequality is by using \cref{eq:def_univariate_restriction_phi} to equivalently rewrite \cref{ineq:grad_xeta_inner_xtwo_minus_xeta}, and the equality is by construction of $t_{\mathrm{ac}}$ in \cref{def:x_ac} and by convexity of $\xi$. Since the univariate function $\xi$ is strictly convex, its derivatives are strictly monotone; hence, \cref{eq:psi_prime_zero_negative} implies that $t_{\mathrm{ac}}\geq 0$. Recalling  that by \cref{eq:def_univariate_restriction_phi}, $t=0$ corresponds to $\vx^\eta$ and $t=1$ corresponds to $\vx_2$, and that  $\vx_{\mathrm{ac}} \in \mathrm{int}(\mathcal{S} \cap \mathcal{K})$, we may deduce from   $t_{\mathrm{ac}}\geq0$ that 
\[\vx_{\mathrm{ac}} \in \{\vx^\eta t + (1 - t)\vx_2 \mid t \in [0,1]\}.\numberthis\label[ineq]{ineq:vxac_between_vxtwo_and_vxeta}\] Combining \cref{ineq:xac_xtwo_lb,ineq:vxac_between_vxtwo_and_vxeta}, we have  \[\|\vx^\eta - \vx_2\| \geq \|\vx_{\mathrm{ac}} - \vx_2\| \geq \frac{r}{2\nu + 4\sqrt{\nu}},\numberthis\label[ineq]{ineq:x_eta_lb_eta_first_part}\] which proves one part of \cref{ineq:dist_to_xeta_xo_xt}. 

\paragraph{Lower bound on $\|\vx^\eta-\vx_1\|$. }
We now parameterize $\mathcal{S}$ by $\psi(t) = \vx_1 + t\frac{\bar{\vx} - \vx_1}{\|\bar{\vx} - \vx_1\|}$.  Define  $c := \nabla q(\vx^\eta)^\top \left(\frac{\bar{\vx} - \vx_1}{\|\bar{\vx} - \vx_1\|}\right)$.   
 We then define the following two optimizers
\begin{align*}
t^\star := \arg\min_{t, \psi(t) \in \mathcal{K}} c \cdot t, \text{ and } t^\eta := \arg\min_{t} c\cdot t + \eta \phi(\psi(t)).
\end{align*}
It follows from  $\psi([0,2r]) \subseteq \mathcal{K}$ and $c\geq0$ (by our choice of $\vx_1$) that $t^\star=0$.   
We then apply \Cref{lem:zong_opt_gap} with the above $c$, $t^\eta$, $t^\star$, and barrier $\phi(\psi({}\cdot{}))$ (with its associated self-concordance parameter $\nu_{\phi\circ\psi}\leq \nu$~\cite{renegar2001mathematical}) and Cauchy-Schwarz inequality to conclude that,
\begin{align*}
    \min\left\{\frac{\eta}{2\|\nabla q(\vx^\eta)\|}, \frac{r}{2\nu + 4 \sqrt{\nu}}\right\} \leq \min\left\{\frac{\eta}{2c}, \frac{r}{2\nu + 4 \sqrt{\nu}}\right\}\leq t^\eta. \numberthis\label[ineq]{ineq:x_eta_lb_eta_second_part}
\end{align*} Finally, note that the optimality condition of $\vx^\eta$ implies $\nabla q(\vx^\eta) + \eta \nabla \phi(\vx^\eta) = 0$, and specifically that $[\nabla q(\vx^\eta) + \eta \nabla \phi(\vx^\eta) ]^\top \left(\frac{\bar{\vx} - \vx_1}{\|\bar{\vx} - \vx_1\|}\right) = 0$. Since $\vx^\eta \in \mathcal{S}$, we can write 
\begin{align*}
    0 &= [\nabla q(\vx^\eta) + \eta \nabla \phi(\vx^\eta) ]^\top \left(\frac{\bar{\vx} - \vx_1}{\|\bar{\vx} - \vx_1\|}\right) \\
    &= \nabla q(\vx^\eta)^\top\left(\frac{\bar{\vx} - \vx_1}{\|\bar{\vx} - \vx_1\|}\right) + \eta \frac{d(\phi \circ \psi)}{dt} (\psi^{-1}(\vx^\eta))) \\
    &= c + \eta \frac{d(\phi \circ \psi)}{dt} (\psi^{-1}(\vx^\eta)).
\end{align*}
We can observe that $c + \eta \frac{d}{dt} (\phi \circ \psi)(t)\vert_{t=t^\eta} = 0$ is the optimality condition of $t^\eta$. Since $\phi$ and, by extension, $\phi \circ \psi$ are strongly convex, we have that $\psi^{-1}(\vx^\eta) = t^\eta$. Since $\psi$ is parameterized in terms of distance from $\vx_1$, we have that,
\[t^\eta=\|\vx^\eta - \vx_1\|. \numberthis\label{eq:t_eta_x_eta_minus_x_one}\] Combining \cref{ineq:x_eta_lb_eta_first_part,ineq:x_eta_lb_eta_second_part,eq:t_eta_x_eta_minus_x_one} yields \cref{{ineq:dist_to_xeta_xo_xt}}, which, as argued earlier, concludes the proof of the lemma.  
\end{proof}

\noindent Similar to \Cref{lem:iso_quad_upper_bound}, we now adapt this to get a lower bound for isotropic quadratics.

\begin{lemma}\label{lem:iso_quad_ball} Let $\mathcal{K} \subset \R^n$ be a polytope and $\phi$ be a $\nu$-self-concordant barrier function. Assume there exists $\bar{\vx} \in \mathcal{K}$ such that $\mathcal{B}(\bar{\vx},r) \subseteq \mathcal{K} \subseteq \mathcal{B}(\bar{\vx}, R)$ for some $r, R > 0$. Let $\vx^\eta := \argmin_{\vx} \frac{\alpha}{2}\|\vx - \vv\|^2 + \eta \phi(\vx)$, $\vx^\star := \argmin_{\vx \in \mathcal{K}} \frac{\alpha}{2}\|\vx - \vv\|^2$ for some $\vv \in \R^n$. Then we know the following ball centered around $\vx^\eta$ is contained within $\mathcal{K}$. 
\begin{align*}
    \mathcal{B}\left(\vx^\eta, \frac{r}{R} \min\left\{\frac{1}{2\sqrt{\nu}}\left(\sqrt{\frac{\eta}{\alpha} + \|\vx^\star - \vv\|^2} - \|\vx^\star -\vv\|\right), \frac{r}{2\nu + 4 \sqrt{\nu}}\right\}\right) \subseteq \mathcal{K}.
\end{align*}
\end{lemma}
\begin{proof}To prove this we use \Cref{lem:contains_linear_ball} and techniques similar to \Cref{lem:iso_quad_upper_bound}. By the triangle inequality and \Cref{lem:iso_quad_universal_upper_bound}, we have 
\begin{align*}
    \alpha \|\vx^\eta - \vv\| 
    \leq \alpha (\|\vx^\eta - \vx^\star\| + \|\vx^\star - \vv\|) 
    \leq \sqrt{\alpha}\sqrt{\eta\nu} + \alpha \|\vx^\star  - \vv\|.
\end{align*}
We next apply this upper bound to the result of \cref{lem:contains_linear_ball} with $q(\vx)= \frac{\alpha}{2}\|\vx-\vv\|^2$ and obtain: 
\begin{align*}
    \mathcal{B}\left(\vx^\eta, \frac{r}{R} \min\left\{\frac{\eta}{2(\sqrt{\alpha \eta}\sqrt{\nu} + \alpha \|\vx^\star -\vv\|)}, \frac{r}{2\nu + 4 \sqrt{\nu}}\right\}\right) \subseteq \mathcal{K}.
\end{align*}
With some rearranging, we may express the above bound as:
\begin{align*}
    \mathcal{B}\left(\vx^\eta, \frac{r}{R} \min\left\{\frac{1}{2\sqrt{\nu}}\frac{\frac{\eta}{\alpha} }{\sqrt{\frac{\eta}{\alpha}} + \frac{1}{\sqrt{\nu}}\|\vx^\star - \vv\|}, \frac{r}{2\nu + 4 \sqrt{\nu}}\right\}\right) \subseteq \mathcal{K}.
\end{align*}
Observe that for any $x > 0, y \in \R$, we have that,
\begin{align*}
    \sqrt{x + y^2} - y = \frac{x}{\sqrt{x + y^2} + y} \leq \frac{x}{\sqrt{x} + y}.
\end{align*}
Since $\nu \geq 1$, we have $\frac{\eta/\alpha}{\sqrt{\eta/\alpha} + \frac{1}{\sqrt{\nu}} \|\vx^\star - \vv\|} \geq \frac{\eta/\alpha}{\sqrt{\eta/\alpha} +  \|\vx^\star - \vv\|} \geq \sqrt{\frac{\eta}{\alpha} + \|\vx^\star - \vv\|^2} - \|\vx^\star - \vv\|$. We can then simplify the last bound on the radius around $\vx^\eta$ to match the claimed bound. 
\end{proof}

\subsubsection{Consolidated Upper and Lower Bounds}

We now collect \Cref{lem:iso_quad_upper_bound} and \Cref{lem:iso_quad_ball}, performing a change of basis to provide bounds for arbitrary quadratic objective functions.

\begin{theorem}\label{thm:quad_opt_result}Let $\mathcal{K} = \{\vx : \mA \vx \geq \vb\}$ be a polytope for some $\mA \in \R^{n_r \times n}, \vb \in \R^{n_r}$. Let $\phi$ be a $\nu$-self-concordant barrier over $\mathcal{K}$. Assume there exists $\bar{\vx} \in \mathcal{K}$ such that $\mathcal{B}(\bar{\vx},r) \subseteq \mathcal{K} \subseteq \mathcal{B}(\bar{\vx}, R)$ for some $r, R > 0$. Let,
\begin{align*}
\vx^\eta &:= \argmin_{\vx} \frac{1}{2}(\vx - \vv)^\top \mH (\vx - \vv) + \eta \phi(\vx), \\
\vx^\star &:= \argmin_{\vx \in \mathcal{K}} \frac{1}{2}(\vx - \vv)^\top \mH (\vx - \vv),
\end{align*}
where $m\mI \preceq \mH \preceq M\mI$. Let $\va = \frac{\mH(\vx^\star - \vv)}{\|\mH(\vx^\star - \vv)\|}$ if $\|\vx^\star - \vv\| > 0$. Then the following hold: 
\begin{thmenum}
\compresslist{
\item\label[theorem]{item:thm_consolidated_i} $    \|\vx^\eta - \vx^\star\| \leq \sqrt{\frac{\eta\nu}{m}}$,
\item\label[theorem]{item:thm_consolidated_ii} $0 \leq \va^\top (\vx^\eta - \vx^\star) \leq \frac{1}{2\sqrt{m}}\left(\sqrt{{4\eta\nu} + \|\vx^\star - \vv\|_{\mH}^2} - \|\vx^\star - \vv\|_{\mH}\right)$,
\item\label[theorem]{item:thm_consolidated_iii} $ \mathcal{B}\left(\vx^\eta, \sqrt{\frac{m}{M}} \cdot \frac{r}{R}\cdot \min\left\{\frac{1}{\sqrt{\nu M}}\left(\sqrt{\eta + \|\vx^\star - \vv\|_{\mH}^2} - \|\vx^\star -\vv\|_{\mH}\right), \sqrt{\frac{m}{M}} \cdot \frac{r}{2\nu + 4 \sqrt{\nu}}\right\}\right) \subseteq \mathcal{K}.$
}
\end{thmenum}
Note that this implies that, if $\va$ exists, then
\begin{align*}
    \sqrt{\frac{m}{M}} \cdot \frac{r}{R} \cdot\min\left\{\frac{1}{\sqrt{\nu M}}\left(\sqrt{\eta + \|\vx^\star - \vv\|_{\mH}^2} - \|\vx^\star -\vv\|_{\mH}\right), \sqrt{\frac{m}{M}} \cdot \frac{r}{2\nu + 4 \sqrt{\nu}}\right\} \leq \va^\top(\vx^\eta - \vx^\star).
\end{align*}
\end{theorem}
\noindent Note that \cref{{item:thm_consolidated_iii}} implies a lower bound on the residual (distance from boundary of $\mathcal{K}$) of $\vueta$. This is the bound we directly use in our smoothness bound in  \Cref{{thm:hess_ueta_bounded}}. 
\begin{proof}[Proof of \Cref{{thm:quad_opt_result}}]For both the upper and lower bounds, we use the change of basis $\vy = \mH^{1/2}\vx$, $\vz = \mH^{1/2}\vv$. We can then transform the assumed definitions of $\vx^\eta$ and $\vx^\star$ into the following optimization problem in $\vy$: 
\begin{align*}
\vy^\eta &:= \argmin_{\vy} \frac{1}{2}\|\vy - \vz\|^2 + \eta \cdot\phi(H^{-1/2}\vy), \\
\vy^\star &:= \argmin_{y\in H^{1/2} \cdot \mathcal{K}} \frac{1}{2}\|\vy - \vz\|^2.
\end{align*}
By \Cref{thm:affine_scb}, we have that $\phi \circ \mH^{-1/2}$ defined on the set $\mH^{1/2}\cdot\mathcal{K}$ is still a self-concordant barrier with parameter $\nu$. Therefore, 
we have that,
\begin{align*}
    \|\vx^\eta - \vx^\star\| \leq \frac{1}{\sqrt{m}}\|\vy^\eta -\vy^\star\| \leq \sqrt{\frac{\eta\nu}{m}},
\end{align*}
where the final step is by \Cref{lem:iso_quad_universal_upper_bound}. This completes the proof of \Cref{item:thm_consolidated_i}. Next, we define  $\tilde{\va} := \frac{\vy^\star - \vz}{\|\vy^\star - \vz\|}$ and $\hat{\va} := \frac{1}{\sqrt{m}} \cdot \mH^{1/2}\tilde{\va}$. 
Then we have $\va = \hat{\va}/\|\hat{\va}\| = \frac{\mH(\vx^\star - \vv)}{\|\mH(\vx^\star - \vv)\|}$ and 
\begin{align*}
    0 \leq \va^\top (\vx^\eta - \vx^\star) = \frac{1}{\|\hat{\va}\|} \hat{\va}^\top(\vx^\eta - \vx^\star) = \frac{1}{\|\hat{\va}\| \sqrt{m}} \tilde{\va}^\top(\vy^\eta - \vy^\star) &\leq \frac{1}{2\sqrt{m}}\left(\sqrt{4\eta\nu + \|\vy^\star - \vz\|^2} - \|\vy^\star - \vz\|\right).
\end{align*}
where the first and final inequalities are by  \Cref{lem:iso_quad_upper_bound} 
applied to $\vy^\eta$ and $\vy^\star$, and using that $\|\hat{\va}\| \geq 1$. Finally, by again using this affine transformation between $y$ and $x$, we obtain \cref{item:thm_consolidated_ii}:
\begin{align*}
    0 \leq \va^\top (\vx^\eta - \vx^\star) \leq \frac{1}{2\sqrt{m}}\left(\sqrt{4\eta\nu + \|\vx^\star - \vv\|_{H}^2} - \|\vx^\star - \vv\|_{H}\right).
\end{align*}
For the lower bound, we note that $\mH^{1/2} \mathcal{K}$ contains a ball of radius at least $r \cdot \sqrt{m}$ and is contained in a ball of radius at most $R\cdot \sqrt{M}$. Applying \Cref{lem:iso_quad_ball}, we have that, 
\begin{align*}
    \mathcal{B}\left(\vy^\eta, \sqrt{\frac{m}{M}} \cdot \frac{r}{R} \min\left\{\frac{1}{\sqrt{\nu}}\left(\sqrt{\eta + \|\vy^\star - \vz\|^2} - \|\vy^\star -\vz\|\right), \sqrt{m} \cdot \frac{r}{2\nu + 4 \sqrt{\nu}}\right\}\right) \subseteq \mH^{1/2} \cdot \mathcal{K}.
\end{align*}
Plugging in the definitions for $\vy$ and $\vz$ in the LHS, we have,
\begin{align*}
    \mathcal{B}\left(\vy^\eta, \sqrt{\frac{m}{M}} \cdot \frac{r}{R} \min\left\{\frac{1}{\sqrt{\nu}}\left(\sqrt{\eta + \|\vx^\star - \vv\|_{H}^2} - \|\vx^\star -\vv\|_{H}\right), \sqrt{m} \cdot \frac{r}{2\nu + 4 \sqrt{\nu}}\right\}\right) \subseteq \mH^{1/2} \cdot \mathcal{K}.
\end{align*}
Since $\|\mH^{1/2}\| \leq \sqrt{M}$, we have $\sigma_{\min}(\mH^{-1/2}) \geq \frac{1}{\sqrt{M}}$. This implies that the ball centered around $\vx^\eta$ in $\mathcal{K}$ is at least $\frac{1}{\sqrt{M}}$ the radius of the ball centered around $\vy^\eta$ in $\mH^{1/2} \cdot \mathcal{K}$. Thus, as claimed in \cref{item:thm_consolidated_iii}, we have,
\begin{align*}
    \mathcal{B}\left(\vx^\eta, \sqrt{\frac{m}{M}} \cdot \frac{r}{R} \min\left\{\frac{1}{\sqrt{\nu M}}\left(\sqrt{\eta + \|\vx^\star - \vv\|_{\mH}^2} - \|\vx^\star -\vv\|_{\mH}\right), \sqrt{\frac{m}{M}} \cdot \frac{r}{2\nu + 4 \sqrt{\nu}}\right\}\right) \subseteq \mathcal{K}.
\end{align*}
Finally, note that $\va := \frac{\mH(\vx^\star - v)}{\|\mH(\vx^\star - v)\|} = \frac{\nabla q(\vx^\star)}{\|\nabla q(\vx^\star)\|}$, where $q(\vx) := \frac{1}{2}(\vx - \vv)^\top \mH(\vx - \vv)$ is the objective for the hard-constrained problem. By the optimality of $\vx^\star$ for this problem, $\va$ is a combination of the active constraints at $\vx^\star$, meaning that $$\{\vx : \va^\top(\vx - \vx^\star) \geq 0\} \supseteq \mathcal{K}.$$
Since, by \cref{item:thm_consolidated_iii}, there exists a minimum-radius ball around $\vx^\eta$ (which in turn implies that the minimum distance of $\vx^\eta$ from the boundary of $\mathcal{K}$ is at least this radius) and $\va$ is unit-norm, we can conclude that the distance of $\vx^\eta$ from the hyperplane specified by $a$ is at least this radius: 
\begin{align*}
    \sqrt{\frac{m}{M}} \cdot \frac{r}{R} \min\left\{\frac{1}{\sqrt{\nu M}}\left(\sqrt{\eta + \|\vx^\star - \vv\|_{\mH}^2} - \|\vx^\star -\vv\|_{\mH}\right), \sqrt{\frac{m}{M}} \cdot \frac{r}{2\nu + 4 \sqrt{\nu}}\right\} \leq \va^\top(\vx^\eta - \vx^\star).
\end{align*}
\end{proof}

\end{document}